\documentclass[a4paper,onecolumn,11pt,unpublished]{quantumarticle}
\pdfoutput=1

\usepackage[utf8]{inputenc}
\usepackage[english]{babel}
\usepackage[T1]{fontenc}
\usepackage{amsmath}
\usepackage{amsthm}
\usepackage{amssymb}
\usepackage{hyperref}
\usepackage{braket}
\usepackage[backend=biber,sortcites=true, maxcitenames=2, maxbibnames=99]{biblatex}
\usepackage{booktabs}
\usepackage{csquotes}
\usepackage{enumitem}
\usepackage{float}
\usepackage{graphicx}
\usepackage{hyperref}
\usepackage{mathdots}
\usepackage{mathtools}
\usepackage[linewidth=1pt]{mdframed}
\usepackage{microtype}
\usepackage{multicol}
\usepackage{physics}
\usepackage{relsize}
\usepackage{rotating}
\usepackage{stmaryrd}   
\usepackage{thmtools}
\usepackage{thm-restate}
\usepackage{tikzit}
\usetikzlibrary{decorations.pathreplacing}
\usepackage{todonotes}
\usepackage[compact]{titlesec}
\usepackage{upgreek}
\usepackage[normalem]{ulem}
\usepackage{xargs} 
\usepackage{xcolor}
\usepackage{xfrac}
\allowdisplaybreaks
\newcounter{todocounter}  \newcommandx{\todocount}[2][1=]{\stepcounter{todocounter}\todo[linecolor=YellowGreen,backgroundcolor=YellowGreen!25,bordercolor=YellowGreen,#1]{\thetodocounter: #2}}

\newcommandx{\unsure}[2][1=]{\todo[linecolor=red,backgroundcolor=red!25,bordercolor=red,#1]{\small #2}}
\newcommandx{\change}[2][1=]{\todo[linecolor=blue,backgroundcolor=blue!25,bordercolor=blue,#1]{{\scriptsize #2}}}
\newcommandx{\info}[2][1=]{\todo[linecolor=yellow,backgroundcolor=yellow!25,bordercolor=yellow,#1]{\small #2}}
\newcommandx{\dothis}[2][1=]{\todo[linecolor=purple,backgroundcolor=purple!25,bordercolor=purple,#1]{#2}}

\renewcommand{\vec}[1]{\underline{#1}}
\newcommand{\defeq}{\vcentcolon=}

\def\fcmp{\mathbin{\raise 0.6ex\hbox{\oalign{\hfil$\scriptscriptstyle      \mathrm{o}$\hfil\cr\hfil$\scriptscriptstyle\mathrm{9}$\hfil}}}}

\newcommand{\rvdots}{\rotatebox{90}{...}}

\newcommand{\naturals}{\mathbb{N}} 
\newcommand{\integers}{\mathbb{Z}} 
\newcommand{\reals}{\mathbb{R}} 
\newcommand{\complexs}{\mathbb{C}} 
\newcommand{\integersMod}[1]{\mathbb{Z}_{#1}} 
\newcommand{\nonstd}[1]{\!\,^\star #1}
\newcommand{\starNaturals}{\nonstd{\naturals}} 
\newcommand{\starIntegers}{\nonstd{\integers}} 
\newcommand{\starComplexs}{\nonstd{\complexs}} 
\newcommand{\starReals}{\nonstd{\reals}} 

\newcommand{\suchthat}[2]{\left\{#1\,\middle|\,#2\right\}}

\newcommand{\stdpartSym}{\operatorname{st}}
\newcommand{\stdpart}[1]{\stdpartSym\left(#1\right)}

\newcommand{\starIntegersMod}[1]{{\nonstd{\integersMod{#1}}}}

\newcommand{\LtwoRn}{L^2\left[\reals^n\right]}
\newcommand{\LtwoR}{L^2\left[\reals\right]}
\newcommand{\LatticeR}{\mathfrak{R}}
\newcommand{\LtwoRNonstd}{\starComplexs\left[\LatticeR\right]}

\newcommand{\HilbCategory}{\operatorname{Hilb}} 
\newcommand{\fHilbCategory}{\operatorname{fHilb}} 
\newcommand{\starHilbCategory}{{}^\star\!\fHilbCategory} 

\newcommand{\AffLagRel}{\operatorname{AffLagRel}} 
\newcommand{\GSACategory}{\operatorname{GSA}} 
\newcommand{\ZXGaussCategory}{\operatorname{ZX_G}} 

\newtheorem{Th}{Theorem}[section]
\newtheorem{theorem}[Th]{Theorem}
\newtheorem{proposition}[Th]{Proposition}
\newtheorem{corollary}[Th]{Corollary}
\newtheorem{lemma}[Th]{Lemma}
\newtheorem{example}[Th]{Example}

\newcommand{\tikzfigscale}[2]{\ensuremath{\vcenter{\hbox{\scalebox{#1}{$\input{#2.tikz}$}}}}}

\definecolor{zx_grey}{RGB}{211,211,211}
\definecolor{zx_red}{RGB}{232,165,165}
\definecolor{zx_green}{RGB}{216,248,216}

\newcommand{\interp}[1]{\left\llbracket#1\right\rrbracket}

\newcommand{\tikzrefsize}[1]{\scriptsize{#1}}

\newcommand{\lemref}[1]{{\tikzrefsize{\eqref{#1}}}}

\newcommand{\minu}{\texttt{-}}
\newcommand{\plus}{\texttt{+}}

\newcommand{\namedLabel}[1]{\tag{\textsc{\footnotesize #1}}\label{rule:#1}\refstepcounter{equation}}
\newcommand{\ruleref}[1]{\textsc{\tiny (\hyperref[rule:#1]{#1})}}
\newcommand{\textruleref}[1]{\textsc{(\hyperref[rule:#1]{#1})}}

\newcommand{\vdotsn}{{\ \tikz[baseline, every node/.style={inner sep=0}]{\node at (0,.23){$\vdots$}; \node at (.35,0){$\scriptstyle n$}}\!}}
\newcommand{\vdotsm}{{\ \tikz[baseline, every node/.style={inner sep=0}]{\node at (0,.23){$\vdots$}; \node at (.35,0){$\scriptstyle m$}}\!}}

\renewcommand{\sin}{\textsf{sin}}

\renewcommand{\tan}{\textsf{tan}}
\renewcommand{\csc}{\textsf{csc}}

\renewcommand{\cot}{\textsf{cot}}

\tikzstyle{gn}=[draw=black, shape=circle, fill={zx_green}, draw=black, inner sep=0.7mm, minimum width=0pt, minimum height=0pt, tikzit fill={rgb,255: red,181; green,215; blue,181}]
\tikzstyle{rn}=[gn, fill={zx_red}, draw=black, tikzit fill={rgb,255: red,215; green,96; blue,96}]
\tikzstyle{fn}=[gn, fill={rgb,255: red,255; green,220; blue,185}, tikzit fill={rgb,255: red,255; green,220; blue,185}, tikzit draw=black]
\tikzstyle{gn_phase}=[shape=rectangle, fill={zx_green}, draw=black, minimum size=1.2em, rounded corners=0.5em, inner sep=0.2em, outer sep=-0.2em, scale=0.8, font={\footnotesize}, tikzit shape=circle, tikzit fill={rgb,255: red,181; green,215; blue,181}]
\tikzstyle{rn_phase}=[{gn_phase}, fill={zx_red}, draw=black, tikzit fill={rgb,255: red,215; green,96; blue,96}]
\tikzstyle{fn_phase}=[{gn_phase}, fill={rgb,255: red,255; green,220; blue,185}, tikzit fill={rgb,255: red,240; green,207; blue,174}, tikzit draw=black]
\tikzstyle{lsplit}=[shape=isosceles triangle, isosceles triangle stretches=true, fill=white, draw=black, minimum width=3.33mm, minimum height=2.5mm, inner sep=1pt, outer sep=0mm, shape border rotate=180]
\tikzstyle{rmerge}=[shape=isosceles triangle, isosceles triangle stretches=true, fill=white, draw=black, minimum width=3.33mm, minimum height=2.5mm, inner sep=1pt, outer sep=0mm]
\tikzstyle{vsplit}=[shape=isosceles triangle, isosceles triangle stretches=true, fill=white, draw=black, minimum width=3.33mm, minimum height=2.5mm, inner sep=1pt, shape border rotate=90]
\tikzstyle{lmat}=[shape=signal, signal to=west, signal from=east, fill={zx_grey}, draw=black, minimum height=6pt, inner sep=0.5pt, font={\scriptsize}, tikzit fill=gray, tikzit category=GLA]
\tikzstyle{rmat}=[lmat, shape=signal, signal to=east, signal from=west, tikzit fill=gray, tikzit category=GLA]
\tikzstyle{dmat}=[lmat, shape=signal, signal to=west, signal from=east, tikzit fill=gray, tikzit category=GLA, rotate=90]
\tikzstyle{umat}=[lmat, shape=signal, signal to=east, signal from=west, tikzit fill=gray, tikzit category=GLA, rotate=90]
\tikzstyle{hbox}=[box, draw=black, shape=rectangle, fill=yellow, minimum size=.67em, font={\tiny}]
\tikzstyle{th gn}=[draw=black, shape=circle, fill={zx_green}, draw=black, inner sep=0.8mm, minimum width=0pt, minimum height=0pt, tikzit fill={rgb,255: red,181; green,215; blue,181}, line width=1pt]
\tikzstyle{th gn_phase}=[shape=rectangle, fill={zx_green}, draw=black, minimum size=1.2em, rounded corners=0.5em, inner sep=0.2em, outer sep=-0.2em, scale=0.8, font={\footnotesize}, tikzit shape=circle, tikzit fill={rgb,255: red,181; green,215; blue,181}, line width=1pt]
\tikzstyle{th rn}=[th gn, fill={zx_red}, draw=black, tikzit fill={rgb,255: red,215; green,96; blue,96}]
\tikzstyle{th rn_phase}=[{th gn_phase}, fill={zx_red}, draw=black, tikzit fill={rgb,255: red,215; green,96; blue,96}]
\tikzstyle{box}=[draw=black, shape=rectangle, fill=white, minimum size=.70em, inner sep=0.10em, scale=0.85, font={\small}]
\tikzstyle{th rmerge}=[shape=isosceles triangle, isosceles triangle stretches=true, fill=white, draw=black, minimum width=3mm, minimum height=2.4mm, inner sep=1pt, outer sep=0mm, line width=1pt]
\tikzstyle{th lsplit}=[th rmerge, shape border rotate=180]
\tikzstyle{th dsplit}=[th rmerge, shape border rotate=90]
\tikzstyle{small box}=[fill=white, draw=black, shape=rectangle, minimum width=0.75cm, minimum height=0.75cm]
\tikzstyle{long box}=[fill=white, draw=black, shape=rectangle, minimum width=0.75cm, minimum height=1.25cm]
\tikzstyle{very long box}=[fill=white, draw=black, shape=rectangle, minimum width=0.75cm, minimum height=1.75cm]
\tikzstyle{rembed}=[draw, rounded rectangle, rounded rectangle east arc=0pt, font={\footnotesize}, inner sep=2pt, minimum height=0.225cm, minimum width=0.27cm]
\tikzstyle{lembed}=[draw, rounded rectangle, rounded rectangle west arc=0pt, font={\footnotesize}, inner sep=2pt, minimum height=0.225cm, minimum width=0.27cm]
\tikzstyle{white square}=[fill=white, draw=black, shape=rectangle, font={\LARGE}]
\tikzstyle{left_text}=[anchor=west, tikzit shape=circle]
\tikzstyle{vacuum}=[th gn, fill=white, scale=0.8]
\tikzstyle{graphv}=[fill=white, draw=none, shape=circle, inner sep=0mm]

\tikzstyle{braceedge}=[-, decorate, decoration={brace, amplitude=2mm, raise=-1mm}]
\tikzstyle{blue-edge}=[-, very thick, draw=blue]
\tikzstyle{hadamard edge}=[-, dashed, draw=blue]
\tikzstyle{dashed}=[-, dash pattern=on 0.8mm off 0.8mm]
\tikzstyle{blue}=[-, draw=blue]
\tikzstyle{thick}=[-, line width=1pt]
\tikzstyle{dotsedge}=[-, dotted, decoration={brace, amplitude=2mm, raise=-1mm}]

\begin{document}

\title{The Focked-up ZX Calculus:\newline Picturing Continuous-Variable Quantum Computation}
\author{Razin A. Shaikh}
\affiliation{Department of Computer Science, University of Oxford}
\author{Lia Yeh}
\affiliation{Department of Computer Science, University of Oxford}
\author{Stefano Gogioso}
\affiliation{Department of Computer Science, University of Oxford}
\affiliation{Hashberg Ltd.}

\maketitle

\begin{abstract}
    While the ZX and ZW calculi have been effective as graphical reasoning tools for finite-dimensional quantum computation, the possibilities for continuous-variable quantum computation (CVQC) in infinite-dimensional Hilbert space are only beginning to be explored.
In this work, we formulate a graphical language for CVQC.
Each diagram is an undirected graph made of two types of spiders: the Z spider from the ZX calculus defined on the reals, and the newly introduced Fock spider defined on the natural numbers.
The Z and X spiders represent functions in position and momentum space respectively, while the Fock spider represents functions in the discrete Fock basis.
In addition to the Fourier transform between Z and X, and the Hermite transform between Z and Fock, we present exciting new graphical rules capturing heftier CVQC interactions.

We ensure this calculus is complete for all of Gaussian CVQC interpreted in infinite-dimensional Hilbert space, by translating the completeness in affine Lagrangian relations by Booth, Carette, and Comfort.
Applying our calculus for quantum error correction, we derive graphical representations of the Gottesman-Kitaev-Preskill (GKP) code encoder, syndrome measurement, and magic state distillation of Hadamard eigenstates.
Finally, we elucidate Gaussian boson sampling by providing a fully graphical proof that its circuit samples submatrix hafnians.
\end{abstract}

\section{Introduction}
\label{sec:introduction}

The ZX calculus is a graphical language originally developed for researching qubit quantum computation~\cite{ coeckeInteractingQuantumObservables2008}.
It has found practical applications in many areas, including quantum circuit compilation and optimization~\cite{duncanGraphtheoreticSimplificationQuantum2020, vandeweteringOptimalCompilationParametrised2024}, quantum programming languages~\cite{borgnaEncodingHighlevelQuantum2023}, quantum error correction~\cite{huangGraphicalCSSCode2023, townsend-teagueFloquetifyingColourCode2023},  classical simulation~\cite{codsiClassicallySimulatingQuantum2023, sutcliffeProcedurallyOptimisedZXDiagram2024}, measurement-based quantum computation~\cite{duncanRewritingMeasurementBasedQuantum2010,  mcelvanneyFlowpreservingZXcalculusRewrite2023}, and quantum education~\cite{dundar-coeckeQuantumPicturalismLearning2023}.
Variants of the ZX calculus have also been used to study Hamiltonians in quantum chemistry~\cite{shaikhHowSumExponentiate2022}, analyze ansatz for quantum machine learning~\cite{wangDifferentiatingIntegratingZX2022}, and describe symmetries in condensed matter physics~\cite{eastAKLTStatesZXDiagramsDiagrammatic2022}.

These graphical languages have been formulated for qudits (higher-dimensional systems)~\cite{poorCompletenessArbitraryFinite2023,royQuditZHCalculusGeneralised2023, debeaudrapSimpleZXZH2023} and advanced qudit quantum compilation~\cite{poorQupitStabiliserZXtravaganza2023}.
The past year marks the milestone that quantum graphical language completeness was achieved for all finite dimensional Hilbert spaces and all linear maps between them~\cite{poorCompletenessArbitraryFinite2023, wangCompletenessQufiniteZXW2024, devismeMinimalityFiniteDimensionalZWCalculi2024, poorZXcalculusCompleteFiniteDimensional2024}.
Despite the success of the ZX calculus for qubit and qudit quantum computing, it could have been the case that these techniques were incompatible with infinite-dimensional settings --- if it were to be a different beast entirely.
Prior to this work, quantum graphical calculi cast in Fock space have restricted to either discrete finite-photon-number systems~\cite{clementLovCalculus2022, defeliceQuantumLinearOptics2022, defeliceLightMatterInteractionZXW2023, heurtel2024complete}, or the efficiently classically simulable Gaussian fragment~\cite{boothCompleteGaussian2024}.

Continuous-variable quantum computing (CVQC)~\cite{loydContinuousVariables1999, braunsteinContinuousVariablesReview2005} is a promising paradigm of quantum computation, where the continuous degrees of freedom are used to encode quantum information.
Reconciling CVQC with techniques developed for finite-dimensional quantum computation brings its own set of challenges.
Qubit programming languages can build upon conventional classical computing data structures, with the simplest example being that qubit circuits can be na\"ively and inefficiently simulated through matrix multiplications.
In contrast, even representing continuous-variable quantum computations is challenging, as evaluating them often yields integrals without closed-form solutions.
This presents issues for verifying quantum programs or experiment outputs.
Current CVQC simulation software libraries~\cite{killoran_strawberry_2019, kolarovszki_piquasso_2024} are limited to three special cases: the Gaussian fragment, which is efficiently classical simulable~\cite{bartlett_efficient_2002}; Fock space truncation, which enforces a cutoff of the maximum number of photons in the system; and boson sampling~\cite{aaronson2011bosonsampling} which is non-universal but believed classically hard to simulate.
In the absence of a formal language of more general mathematical descriptors, unifying and having expressibility beyond these three special cases, software design for CVQC is fragmented.

An important open question is how to develop a theoretical software-compatible framework partitioning computations into inhomogeneous parts, to which efficient optimization routines, approximations, or numerical methods may be applied.
Reconstructing CVQC in a more expressive and rigorous formalism facilitates understanding, software implementation, and the development of new methods.
To this end, we formulate a powerful graphical language, which we hope will serve as a starting point for future programming language specifications for CVQC.
Combined with complete ZX calculus rules for Gaussians building upon Ref.~\cite{boothCompleteGaussian2024}, we add rules for non-Gaussian computations making it possible to compile, simulate, and verify interesting CVQC computations without having to evaluate or symbolically manipulate multi-dimensional integrals with infinite domain.
We further believe that the ability to transfer techniques well understood in the qubit or qudit ZX calculus --- such as any from error correction, circuit optimization, or stabilizer decompositions --- will be a game-changer for CVQC.

The generators of our graphical language are Z spiders, corresponding to position space; and Fock spiders, corresponding to the Fock (or number) basis.  We can in turn use the Z and Fock spiders to define X spiders, corresponding to momentum space.
The interactions between these continuous and discrete bases contribute to a highly expressive and powerful language for CVQC.
Furthermore, we define other useful gadgets such as the W node, which with the Fock spider enhance this calculus with rules generalizing the qubit and qudit ZW calculus~\cite{coeckeThreeQubitEntanglement2011, hadzihasanovicAlgebraEntanglementGeometry2017, devismeMinimalityFiniteDimensionalZWCalculi2024} to the Fock basis.
We provide rigorous semantics for our calculus via non-standard analysis in Appendix~\ref{sec:non-standard}, enabling us to reason diagrammatically about infinite-dimensional systems in a dagger-compact context.

An independent work developing ZX calculus for continuous-variable quantum computation was recently put forth by Nagayoshi et.~al.~\cite{nagayoshiZXGraphicalCalculus2024}.
This differs significantly from the present work.
The foremost difference is that the present work is the only known graphical formulation of CVQC which includes both the continuous picture and the Fock space picture.
Ref.~\cite{nagayoshiZXGraphicalCalculus2024} introduces a restriction of the continuous picture, where Z and X spiders are labelled by phase functions fixed by real polynomials. In the present work, Z and X spiders are labelled by arbitrary functions from $\reals\to\complexs$.
We introduce the Fock spider and W nodes for CVQC for the first time, and moreover uncover novel interaction rules between the continuous and Fock space pictures.

\paragraph*{Structure of the paper}
We start by introducing the ZX calculus for continuous-variable quantum processes in Section~\ref{sec:zx-calculus}.
We next examine common CVQC operations and gates in this calculus in Section~\ref{sec:gates}.
In Section~\ref{sec:gaussian-completeness}, we prove this calculus is complete for Gaussian states and operations.

The second half of the main text is for diving into two applications of this calculus.
To study fault-tolerance, we deduce key graphical representations in the Gottesman-Kitaev-Preskill (GKP) quantum error correcting code (Section~\ref{sec:gkp}).
Finally, we show this is a natural language to study Gaussian boson sampling by providing a simple graphical derivation of submatrix hafnians in Section~\ref{sec:gbs}.

\section{ZX-calculus}
\label{sec:zx-calculus}
In this section, we introduce the ZX-calculus for continuous variable quantum computation.
All the wires in our diagrams will correspond to the infinite-dimensional separable Hilbert space $L^2[\reals]$.


\subsection{Generators}
Our generators consist of two types of spiders: the Z spider and the Fock spider.
The Z spider is associated with the continuous position basis states $\ket{x}_{x \in \reals}$.
The Fock spider is associated with the discrete Fock (or particle number) basis states $\ket{n}_{n \in \naturals}$.
The Fock basis is related to the continuous basis by the Hermite functions:
\begin{equation}
    \ket{n} = \int \psi_n(x) \ket{x} \, dx = \int {(2^n n! \sqrt{\pi})}^{-\frac{1}{2}} H_n(x) e^{-x^2/2} \ket{x} \, dx 
\end{equation}
where $H_n(x)$ is the $n$-th Hermite polynomial and $\psi_n(x)$ is the $n$-th Hermite function.
In accordance, we adopt the following Fourier transform convention:
\begin{equation}
    \mathcal{F}(f) = \int_{\infty}^{\infty} f(x) e^{-i 2 \pi x p} dp
\end{equation}
The generators of the calculus are the Z spider, the Fock spider, and the global scalar, defined as follows:
\begin{align}
    \namedLabel{Z-Spider}
    \tikzfigscale{1}{z-spider}
    \qquad & \overset{\interp{\cdot}}{\longmapsto} \qquad
    \int f(x) \, \ket{x}^{\otimes b} \bra{x}^{\otimes a} \, dx       \\
    \namedLabel{Fock-Spider}
    \tikzfigscale{1}{fock-spider}
    \qquad & \overset{\interp{\cdot}}{\longmapsto} \qquad
    \sum_{n=0}^{\infty} g(n) \ket{n}^{\otimes b} \bra{n}^{\otimes a} \\
    \namedLabel{Global-Scalar}
    k\qquad
    \qquad & \overset{\interp{\cdot}}{\longmapsto} \qquad
    k
\end{align}
where $f: \reals \to \complexs$, $\,g: \naturals \to \complexs$ and $k \in \complexs$.
Using the Z-spider, we can define a cup and a cap for the space:
\begin{equation}
    \tikzfigscale{1}{cup} \qquad \qquad \tikzfigscale{1}{cap}
\end{equation}
The cup and cap satisfy the snake equations:
\begin{equation}
    \tikzfigscale{1}{yanking}
\end{equation}
This property allows us to freely bend wires in our diagrams.
Moreover, the Z and the Fock spiders have a structural symmetry which allows us to interchange any of the legs:
\begin{align}
    \tikzfigscale{1}{z-symmetry} \\
    \tikzfigscale{1}{f-symmetry}
\end{align}
This is known as ``flexsymmetry'', or by the slogan ``only connectivity matters''.

\subsection{Notations}
\label{sec:notations}
The Fock basis diagonalizes the Fourier transform between the position and the momentum basis.
Hence, we can use it to define the X spider, associated with the momentum basis states $\ket{p}_{p \in \reals}$.
\begin{equation}
    \namedLabel{X-Spider}
    \tikzfigscale{1}{x-spider}
    \qquad \overset{\interp{\cdot}}{\longmapsto} \qquad
    \int f(p) \, \ket{p}^{\otimes b} \bra{p}^{\otimes a} \, dp
\end{equation}
where $f: \reals \to \complexs$, and the Fock spiders labelled by $(-i)^{\hat{n}}$ and $i^{\hat{n}}$ correspond to the Fourier and inverse Fourier transform, respectively.
We will omit the label on the spiders when the corresponding function is the constant function $1$.
\begin{equation}
    \tikzfigscale{1}{phase-free-spiders}
\end{equation}
Using the Z, X and Fock spiders, we can represent the basis states:
\begin{equation}
    \tikzfigscale{1}{ket-x-character}
    \quad \overset{\interp{\cdot}}{\longmapsto} \quad
    \ket{x}
    \qquad \qquad
    \tikzfigscale{1}{ket-p-character}
    \quad \overset{\interp{\cdot}}{\longmapsto} \quad
    \ket{p}
    \qquad \qquad
    \tikzfigscale{1}{fock-basis}
    \quad \overset{\interp{\cdot}}{\longmapsto} \quad
    \ket{n}
\end{equation}
where $\chi_x (p) = e^{-i 2\pi p x}$, $\bar{\chi}_p (x) = e^{i 2\pi p x}$, and $\delta_n$ is the Kronecker delta at $n$.
Next, we define multipliers, labelled by some non-zero $m \in \reals$:
\begin{equation}
    \namedLabel{Multiplier}
    \tikzfigscale{1}{multiplier-definition}
    \qquad \overset{\interp{\cdot}}{\longmapsto} \qquad
    \int \ket{m x}_X \bra{x}_X \, dx
\end{equation}
We also define a multiplier labelled by $m = 0$, as follows:
\begin{equation}
    \tikzfigscale{1}{multiplier-zero-definition}
    \qquad \overset{\interp{\cdot}}{\longmapsto} \qquad
    \int \ket{0}_X \bra{x}_X \, dx
\end{equation}
Finally, we define the W node: this was originally a generator for the ZW calculus, but it can be defined in terms of our existing generators.
We first define the W node with 0, 1, and 2 legs:
\begin{equation}
    \namedLabel{W-Node}
    \tikzfigscale{1}{w-definition}
\end{equation}
Using associativity, we can define the W node with more than 2 legs:
\begin{equation}
    \tikzfigscale{1}{w-node}
\end{equation}
Hence, we have a W-node with $k$ legs, for all $k \in \naturals$.
Its interpretation is given as follows:
\begin{equation}
    \tikzfigscale{1}{w-node-any-legs}
    \qquad \overset{\interp{\cdot}}{\longmapsto} \qquad
    \sum_{n_1, \dots, n_k \geqslant 0}
    \sqrt{\frac{(\sum_i n_i)!}{\prod_i n_i!}}\, \Big|{\sum_i n_i}\Big\rangle \Big\langle {n_1, \dots , n_k}\Big|
\end{equation}
It is also convenient to define the W-node in the opposite direction, by taking the transpose:
\begin{equation}
    \tikzfigscale{1}{w-node-transpose}
    \qquad \overset{\interp{\cdot}}{\longmapsto} \qquad
    \sum_{n_1, \dots, n_k \geqslant 0}
    \sqrt{\frac{(\sum_i n_i)!}{\prod_i n_i!}}\, \Big|{n_1, \dots , n_k}\Big\rangle \Big\langle {\sum_i n_i}\Big|
\end{equation}
Similar to the Z and Fock spiders, the W node also has a structural symmetry which allows legs to be freely interchanged:
\begin{equation}
    \tikzfigscale{1}{w-symmetry}
\end{equation}

\subsection{Rules}

\begin{figure}[!ht]
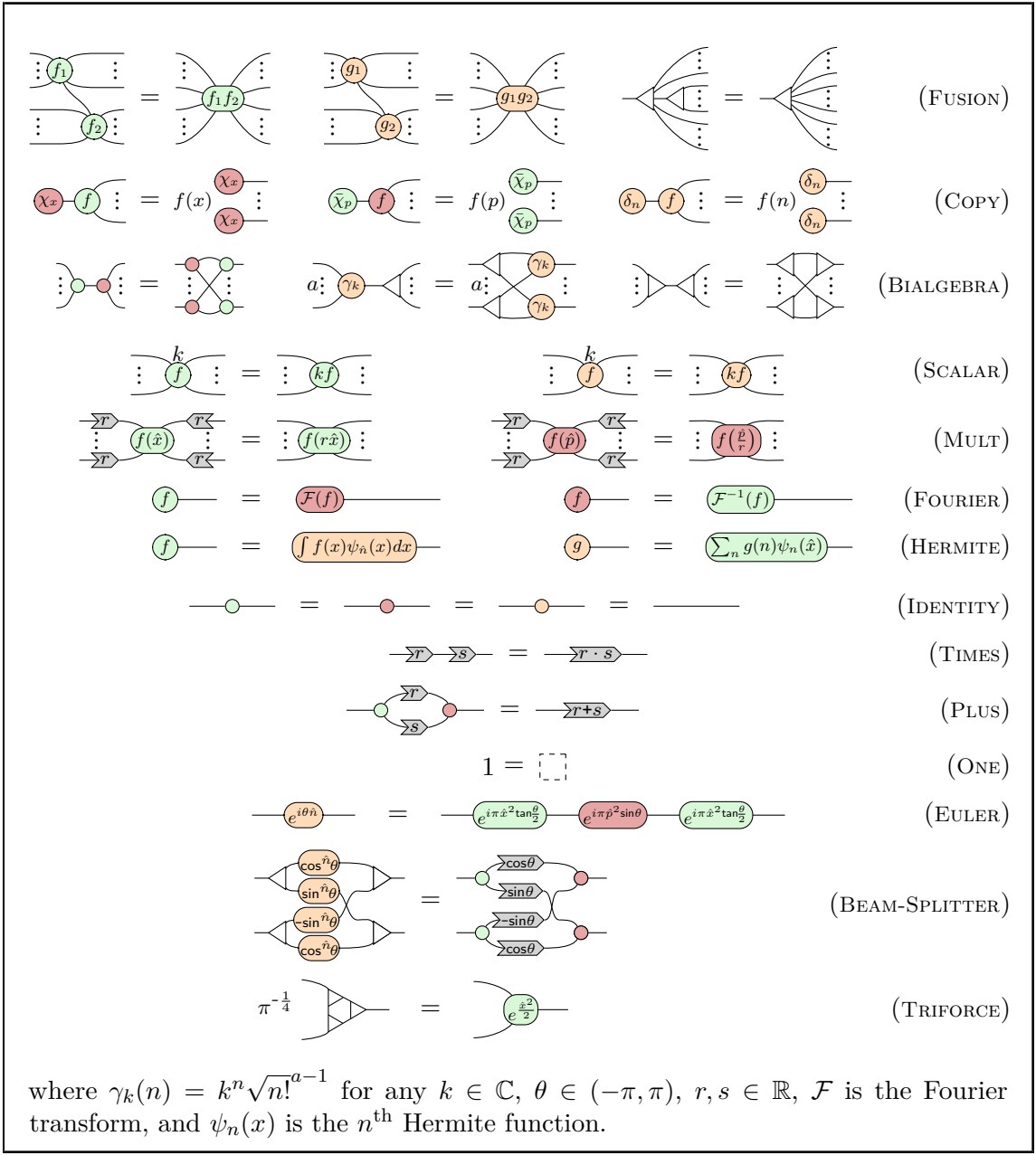

    \begin{mdframed}
        \begin{align}
            \hspace{-1em}\namedLabel{Fusion}
            \tikzfigscale{1}{z-fusion}
             &  & \tikzfigscale{1}{f-fusion}
             &  & \tikzfigscale{1}{w-fusion}         \\[0.3em]
            \hspace{-1em}\namedLabel{Copy}
            \tikzfigscale{1}{z-character-copy}
             &  & \tikzfigscale{1}{x-character-copy}
             &  & \tikzfigscale{1}{f-basis-copy}     \\[0.3em]
            \hspace{-1em}\namedLabel{Bialgebra}
            \tikzfigscale{1}{bialgebra}
             &  & \tikzfigscale{1}{zw-bialgebra}
             &  & \tikzfigscale{1}{w-bialgebra}
        \end{align}
        \vspace{-1\baselineskip}
        \begin{align}
            \namedLabel{Scalar}
            \tikzfigscale{1}{z-scalar}
             &  & \tikzfigscale{1}{f-scalar}              \\[0.3em]
            \namedLabel{Mult}
            \tikzfigscale{1}{z-mult}
             &  & \tikzfigscale{1}{x-mult}                \\[0.3em]
            \namedLabel{Fourier}
            \tikzfigscale{1}{fourier-state}
             &  & \tikzfigscale{1}{inverse-fourier-state} \\[0.3em]
            \namedLabel{Hermite}
            \tikzfigscale{1}{hermite-z-to-f}
             &  & \tikzfigscale{1}{hermite-f-to-z}
        \end{align}
        \vspace{-1\baselineskip}
        \begin{equation}
            \namedLabel{Identity}
            \tikzfigscale{1}{id-spider}
        \end{equation}
        \begin{equation}
            \namedLabel{Times}
            \tikzfigscale{1}{times}
        \end{equation}
        \begin{equation}
            \namedLabel{Plus}
            \tikzfigscale{1}{multiplier-plus}
        \end{equation}
        \begin{equation}
            \namedLabel{One}
            \tikzfigscale{1}{one}
        \end{equation}
        \begin{equation}
            \namedLabel{Euler}
            \tikzfigscale{1}{euler}
        \end{equation}
        \vspace{-1.5\baselineskip}
        \begin{align}
            \namedLabel{Beam-Splitter}
            \tikzfigscale{1}{beam-splitters} \\[0.3em]
            \namedLabel{Triforce}
            \tikzfigscale{1}{triforce}
        \end{align}
        where $\gamma_k(n) = k^n \sqrt{n!}^{a-1}$ for any $k \in \complexs$, $\theta \in (-\pi, \pi)$, $r,s \in \reals$, $\mathcal{F}$ is the Fourier transform, and $\psi_n(x)$ is the $n^\text{th}$ Hermite function.
    \end{mdframed}
    \caption{Rules of the calculus}\label{fig:ZX-Rules}
\end{figure}
The rules of the calculus are given Figure~\ref{fig:ZX-Rules}.
Several rules are analogous to finite-dimensional ZX calculi: the Z and X versions of \textruleref{Fusion}, \textruleref{Copy}, \textruleref{Bialgebra}, \textruleref{Identity}, \textruleref{Times}, \textruleref{Plus}, and \textruleref{One}. These rules were adapted to the infinite-dimensional setting by defining mutually unbiased orthonormal bases for position and momentum using non-standard analysis, and applying a model-theoretic technique known as the Transfer Theorem.
Our model construction is detailed in Appendix~\ref{sec:non-standard}, and soundness of rules is discussed in Appendix~\ref{sec:soundness}.

Other rules have been introduced in the context of quantum graphical calculi for the continuous-variable setting.
The Fock (F) and W rules \textruleref{Fusion}, \textruleref{Bialgebra}, and \textruleref{Identity} will look familiar to those that have encountered the ZW calculus.
However, unlike in previous calculi combining the ZX and ZW calculi~\cite{poorCompletenessArbitraryFinite2023,defeliceLightMatterInteractionZXW2023,shaikhHowSumExponentiate2022}, our Z and F generators do not coincide.
Our Z and X generators are defined over the reals, whereas our F and W generators are defined over the natural numbers.

The most interesting of our new rules is, by far, the \textruleref{Triforce} rule.
This simple rule encapsulates the following product formula of Hermite polynomials~\cite{nielsen1918hermite}:
\begin{equation}
    H_m(x) H_n(x) =
    \sum_{r=0}^{\textsf{min}(m,n)} r!\, 2^r \binom{m}{r} \binom{n}{r} H_{m\texttt{+}n\texttt{-}2r}(x)
\end{equation}
It can further be generalized to arbitrary number of legs, as shown in Proposition~\ref{prop:triforceKn}.

\subsection{Derived rules}
Here, we provide a table of useful rules that can be graphically derived from those in Figure~\ref{fig:ZX-Rules}.
We give the most interesting proofs in this section, and relegate the remaining proofs to the Appendix~\ref{sec:rule-proofs-appendix}.
The derived rules are presented in Figure~\ref{fig:Derived-Rules}.
\begin{figure}[!b]
    \begin{mdframed}
        \begin{multicols}{2}
            \begin{restatable}{lemma}{xFusion}\label{lem:xFusion}
                \begin{equation}
                    \tikzfigscale{1}{x-fusion}
                \end{equation}
            \end{restatable}
            \begin{restatable}{lemma}{nul}\label{lem:nul}
                \begin{equation}
                    \tikzfigscale{1}{null}
                \end{equation}
            \end{restatable}
            \begin{restatable}{lemma}{WUnit}\label{lem:WUnit}
                \begin{equation}
                    \tikzfigscale{1}{w-unit}
                \end{equation}
            \end{restatable}
            \begin{restatable}{lemma}{vacuumCopy}\label{lem:vacuumCopy}
                \begin{equation}
                    \tikzfigscale{1}{w-vacuum-copy}
                \end{equation}
            \end{restatable}
            \begin{restatable}{lemma}{push}\label{lem:push}
                \begin{equation}
                    \tikzfigscale{1}{f-push}
                \end{equation}
            \end{restatable}
            \begin{restatable}{lemma}{wLoopGreen}\label{lem:wLoopGreen}
                \begin{equation}
                    \tikzfigscale{1}{w-loop-green}
                \end{equation}
            \end{restatable}
            \begin{restatable}{lemma}{squeezedVacuum}\label{lem:squeezedVacuum}
                \begin{equation}
                    \tikzfigscale{1}{squeezed-vacuum}
                \end{equation}
            \end{restatable}
            \begin{restatable}{lemma}{KtoNFockZ}\label{lem:KtoNFockZ}
                \begin{equation}
                    \tikzfigscale{1}{k-to-n-fock-z}
                \end{equation}
            \end{restatable}
            \begin{restatable}{lemma}{FWPlusState}\label{lem:FWPlusState}
                \begin{equation}
                    \tikzfigscale{1}{f-plus-state}
                \end{equation}
            \end{restatable}
            \begin{restatable}{proposition}{FWPlus}\label{prop:FWPlus}
                \begin{equation}
                    \tikzfigscale{1}{f-plus}
                \end{equation}
            \end{restatable}
            \begin{corollary}\label{cor:f-plus-gen}
                \begin{equation}
                    \tikzfigscale{1}{f-plus-gen}
                \end{equation}
            \end{corollary}
            \begin{corollary}\label{cor:f-plus-loop}
                \begin{equation}
                    \tikzfigscale{1}{f-plus-loop}
                \end{equation}
            \end{corollary}
        \end{multicols}
    \end{mdframed}
    \caption{Equalities derivable using the rules in Figure~\ref{fig:ZX-Rules}.}
    \label{fig:Derived-Rules}
\end{figure}
\begin{figure}[ht]
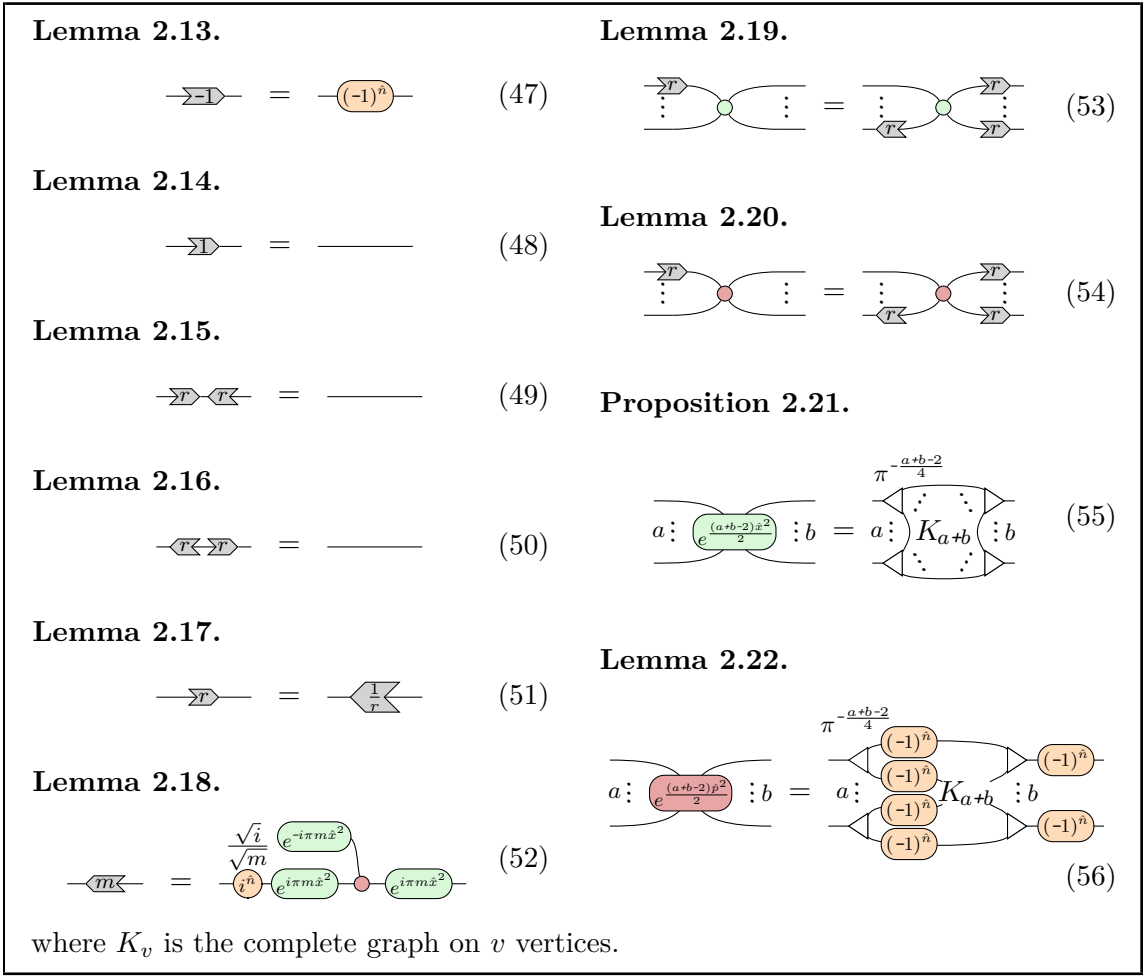
 \ContinuedFloat
    \begin{mdframed}
        \begin{multicols}{2}
            \begin{restatable}{lemma}{minusOne}\label{lem:minusOne}
                \begin{equation}
                    \tikzfigscale{1}{minusOne}
                \end{equation}
            \end{restatable}
            \begin{restatable}{lemma}{plusOne}\label{lem:plusOne}
                \begin{equation}
                    \tikzfigscale{1}{plusOne}
                \end{equation}
            \end{restatable}
            \begin{restatable}{lemma}{multiplierIso}\label{lem:multiplierIso}
                \begin{equation}
                    \tikzfigscale{1}{multiplier-iso}
                \end{equation}
            \end{restatable}
            \begin{restatable}{lemma}{multiplierIsoRev}\label{lem:multiplierIsoRev}
                \begin{equation}
                    \tikzfigscale{1}{multiplier-iso-rev}
                \end{equation}
            \end{restatable}
            \begin{restatable}{lemma}{multiplierRev}\label{lem:multiplierRev}
                \begin{equation}
                    \tikzfigscale{1}{multiplier-rev}
                \end{equation}
            \end{restatable}
            \begin{restatable}{lemma}{multiplierAltDecomposition}\label{lem:multiplierAltDecomposition}
                \begin{equation}
                    \tikzfigscale{1}{multiplier-alt-decomposition}
                \end{equation}
            \end{restatable}

            \begin{restatable}{lemma}{multiplierCopy}\label{lem:multiplierCopy}
                \begin{equation}
                    \tikzfigscale{1}{multiplier-copy}
                \end{equation}
            \end{restatable}
            \begin{restatable}{lemma}{xmultiplierCopy}\label{lem:xmultiplierCopy}
                \begin{equation}
                    \tikzfigscale{1}{x-multiplier-copy}
                \end{equation}
            \end{restatable}
            \begin{restatable}{proposition}{triforceKn}\label{prop:triforceKn}
                \begin{equation}
                    \tikzfigscale{1}{triforce-Kn}
                \end{equation}
            \end{restatable}
            \begin{restatable}{lemma}{xTriforce}\label{lem:xTriforce}
                \begin{equation}
                    \tikzfigscale{1}{x-triforce}
                \end{equation}
            \end{restatable}
        \end{multicols}
        where $K_v$ is the complete graph on $v$ vertices.
    \end{mdframed}
    \caption{Equalities derivable using the rules in Figure~\ref{fig:ZX-Rules} (continued).}
\end{figure}

\triforceKn*
\begin{proof}
    The proof is by induction on $n$, the number of vertices of the complete graph $K_n$.
    The base case given by \textruleref{Triforce} is $n=3$.
    For any $n$, pick any vertex of $K_n$, and compose with it the triforce to get $K_{n+1}$.
    To derive $K_4$ for example:
    \begin{equation*}
        \tikzfigscale{1}{triforce-Kn-proof}
    \end{equation*}
    For all higher $n$, the proof proceeds analogously, due to flexsymmetry of Z and W.
\end{proof}

\section{CVQC gates and operators}
\label{sec:gates}
In this section, we show the translation of a universal set of gates in CVQC to the ZX-calculus.
In analogy with the Clifford+T gateset for qubits, a common universal gateset for CVQC consists of Gaussian operations and a non-Gaussian operation.
We first describe the common operators of CVQC.
Then, we introduce the single qumode Gaussian operations and the entangling Gaussian operations.
Finally, we present the cubic phase gate as a non-Gaussian operation.

\subsection{Operators}
\paragraph*{Creation and annihilation operators}
They can be represented as follows:
\begin{equation*}
    \tikzfigscale{1}{creation}
    \quad \overset{\interp{\cdot}}{\longmapsto} \quad
    a^{\dagger}
    \qquad \text{and} \qquad
    \tikzfigscale{1}{annihilation}
    \quad \overset{\interp{\cdot}}{\longmapsto} \quad
    a
\end{equation*}
\paragraph*{Quadrature operators}
The position and momentum quadrature operators are basic CVQC operators that satisfy $\hat{x} \ket{x} = x \ket{x}$ and $\hat{p} \ket{p} = p \ket{p}$.
They can be represented as the following simple diagrams:
\begin{equation*}
    \tikzfigscale{1}{position-quadrature}
    \quad \overset{\interp{\cdot}}{\longmapsto} \quad
    \hat{x}
    \qquad \text{and} \qquad
    \tikzfigscale{1}{momentum-quadrature}
    \quad \overset{\interp{\cdot}}{\longmapsto} \quad
    \hat{p}
\end{equation*}

\begin{proposition}\label{prop:xhat}
    The (controlled) $\hat{x}$ operator is
    \begin{equation*}
        \tikzfigscale{1}{x-hat}
    \end{equation*}
\end{proposition}
\begin{proof}
    We use the fact that $\hat{x} = \frac{a^\dag \plus a}{\sqrt{2}}$.
    From that expression, we could build a diagram for the operator using the following observation:
    \begin{equation}\label{eq:wsplit}
        \tikzfigscale{1}{w-split}
    \end{equation}
    However, we can also prove the desired equation directly, without using sums of diagrams, thanks to the power of the \textruleref{Triforce}:
    \begin{equation*}
        \tikzfigscale{1}{x-hat-proof}
    \end{equation*}
    Furthermore, we can omit the 1 number state and leave the wire open, obtaining a \emph{controlled} $\hat{x}$ diagram.
    To show that the diagram truly behaves as a controlled $\hat{x}$, we check that inputting the 0 number state would yield the identity instead:
    \begin{equation*}
        \tikzfigscale{1}{x-hat-proof-0}
    \end{equation*}
    Controlled diagrams~\cite{jeandel2018zxrationalangle,shaikhHowSumExponentiate2022} are useful in quantum graphical calculi. This controlled $\hat{x}$ diagram enables drawing Hamiltonians containing this operator.
\end{proof}

\begin{restatable}{proposition}{phat}\label{prop:phat}
    The (controlled) $\hat{p}$ operator is
    \begin{equation*}
        \tikzfigscale{1}{p-hat}
    \end{equation*}
\end{restatable}

\paragraph*{Number operator}
It counts the number of photons in a qumode: $\hat{n} = a^{\dagger} a$ with $\hat{n}\ket{n} = n\ket{n}$.
It corresponds to the following Fock spider:
\begin{equation*}
    \tikzfigscale{1}{number-operator}
    \quad \overset{\interp{\cdot}}{\longmapsto} \quad
    \hat{n}
\end{equation*}
\begin{restatable}{lemma}{fonefact}\label{lem:fonefact}
    \begin{equation*}
        \tikzfigscale{1}{f-1-fact}
    \end{equation*}
\end{restatable}

\begin{proposition}\label{prop:nhat}
    The (controlled) $\hat{n}$ operator is
    \begin{equation*}
        \tikzfigscale{1}{n-hat}
    \end{equation*}
\end{proposition}
\begin{proof}
    We use the fact that $\hat{n} = a^{\dagger}a$ for the case where the 1 state is input:
    \begin{equation*}
        \tikzfigscale{1}{n-hat-proof}
    \end{equation*}
    We check that the diagram reduces to the identity in the case where the 0 state is input:
    \begin{equation*}
        \tikzfigscale{1}{n-hat-proof-0}
    \end{equation*}
\end{proof}

\paragraph*{Heisenberg-Weyl operators}
They are the generalization of the Pauli operators to CVQC.
They displace the qumode in phase space, and they are defined as follows:
\begin{equation}
    X(x) = e^{-i 2\pi x \hat{p}} \qquad \text{and} \qquad Z(p) = e^{i 2\pi p \hat{x}}
\end{equation}
Their action on the position basis states is given by
\begin{equation}
    X(x) \ket{y} = \ket{x+y}\quad \text{ and } \quad Z(p) \ket{y} = e^{i 2\pi p y} \ket{y}.
\end{equation}
The ZX diagrams for the Heisenberg-Weyl operators are spiders labelled by characters $\chi_x(p) = e^{-i 2\pi px}$ and $\bar{\chi}_p(x) = e^{i 2\pi px}$, as follows:
\begin{equation*}
    \tikzfigscale{1}{x-displacement}
    \quad \overset{\interp{\cdot}}{\longmapsto} \quad
    X(x)
    \qquad \text{and} \qquad
    \tikzfigscale{1}{z-displacement}
    \quad \overset{\interp{\cdot}}{\longmapsto} \quad
    Z(p)
\end{equation*}

\subsection{Gaussian operations}
\paragraph*{Displacement}
The displacement operator $D(\alpha)$ shifts the qumode in phase space by $\alpha = x + ip$.
This is equivalent to applying the Heisenberg-Weyl operators $X(x)Z(p)$.
Hence, we can write down the diagram for the displacement operator:
\begin{equation}
    \tikzfigscale{1}{displacement-operator}
    \quad \overset{\interp{\cdot}}{\longmapsto} \quad
    D(x + ip)
\end{equation}
\paragraph*{Squeezing}
The squeeze gate $S(r)$ simply corresponds to a multiplier:
\begin{equation}
    \tikzfigscale{1}{squeezing}
    \quad \overset{\interp{\cdot}}{\longmapsto} \quad
    S(r).
\end{equation}
\paragraph*{Rotation}
The rotation gate that rotates the phase space is $R(\theta) = e^{\minu i \theta \hat{n}}$.
It is representable as a Fock spider and as a decomposition of Z-X-Z spiders, by the \textruleref{Euler} rule:
\begin{equation*}
    \tikzfigscale{1}{minu-euler}
    \quad \overset{\interp{\cdot}}{\longmapsto} \quad
    R(\theta)
\end{equation*}
When $\theta = \frac{\pi}{2}$ and $\minu \frac{\pi}{2}$, this realises the forward and inverse quantum Fourier transforms respectively:
\begin{equation}
    \tikzfigscale{1}{rotation-fourier}
\end{equation}

\paragraph*{Controlled-X gate}
The continuous-variable controlled-X gate applies a controlled displacement on the target qumode.
On the position basis states, it acts as follows:
\begin{equation}
    CX(s) \ket{x} \ket{y} = \ket{x} \ket{y + sx}
\end{equation}
It is represented by the following diagram:
\begin{equation}
    \tikzfigscale{1}{controlled-x}
\end{equation}
\paragraph*{Controlled-phase gate}
The controlled-Z gate can be generalized to CVQC, where it applies a phase:
\begin{equation}
    CZ(s) \ket{x}_X \ket{y}_X = e^{i2\pi sxy} \ket{x}_X \ket{y}_X
\end{equation}
We can obtain this by simply applying fourier and inverse fourier transforms to the target of the controlled-X gate:
\begin{equation}
    \tikzfigscale{1}{controlled-phase}
\end{equation}
When $s=1$, this gate is the usual controlled-Z gate:
\begin{equation}
    \tikzfigscale{1}{controlled-z}
\end{equation}
\paragraph*{Beam splitter}
The beam splitter is a two-mode gate that is easily implemented using linear optics.
The beam splitter and rotation gates allow us to construct arbitrary linear interferometers, which are used in algorithms based on boson sampling.
\begin{equation}
    B(\theta, \phi)\ =\ \exp({\theta (e^{i \phi} a_1 a_2^{\dagger} - e^{-i \phi} a_1^{\dagger} a_2)})
    \quad = \quad
    \tikzfigscale{1}{beam-splitter-fock}
\end{equation}
In the special case where $\phi = 0$, we can represent the beam splitter in the Z-X form by applying the \textruleref{Beam-Splitter} rule:
\begin{equation}
    \tikzfigscale{1}{beam-splitters}
\end{equation}

\subsection{Non-Gaussian operations}
All the gates we have seen so far are Gaussian operations, i.e.\@ their Hamiltonian is quadratic in the creation and annihilation operators.
Non-Gaussian operations involve higher order terms in the Hamiltonian, and they are necessary to enable universal quantum computation.
\paragraph*{Cubic phase gate}
The cubic phase gate is a non-Gaussian operation with a cubic term in the Hamiltonian, represented as follows:
\begin{equation}
    V(\gamma) \quad  = \quad e^{i \gamma \hat{x}^3} \quad = \quad \tikzfigscale{1}{cubic-phase}
\end{equation}
\paragraph*{Kerr gate}
The Kerr gate is another non-Gaussian operation that introduces a quadratic term in the Hamiltonian:
\begin{equation}
    H_{\text{Kerr}} = a^{\dagger} a a^{\dagger} a = \hat{n}^2
\end{equation}
We can represent the Kerr gate using a Fock spider:
\begin{equation}
    K(\kappa) \quad  = \quad e^{i \kappa \hat{n}^2} \quad  = \quad \tikzfigscale{1}{kerr}
\end{equation}
\paragraph*{Cross-Kerr gate}
The cross-Kerr gate is a two-mode gate that performs a non-Gaussian entangling operation.
It has the following graphical representation in our calculus:
\begin{equation}
    CK(\kappa) \quad  = \quad e^{i \kappa \hat{n}_1 \hat{n}_2} \quad  = \quad \tikzfigscale{1}{cross-kerr}
\end{equation}

\section{Completeness for the Gaussian fragment of CVQC}
\label{sec:gaussian-completeness}
In this section, we define the Gaussian fragment of our calculus and prove its completeness.
We prove the completeness by translating to and from another complete graphical calculus, the graphical symplectic algebra~\cite{boothGraphicalSymplecticAlgebra2024,boothCompleteGaussian2024}.
This calculus provides semantics in terms of affine Lagrangian relations~\cite{comfortGraphicalCalculusLagrangian2022}, and was shown to be complete for the Gaussian fragment of CVQC.
Completeness via translation is a common technique in the ZX-calculus literature, and was originally used to prove the completeness of the qubit ZX-calculus~\cite{ngUniversalCompletionZXcalculus2017,jeandelCompleteAxiomatisationZXCalculus2018}.

We first recap the graphical symplectic algebra, then describe the Gaussian fragment of the ZX-calculus, and finally prove the completeness of the Gaussian fragment of the ZX-calculus.

\subsection{Graphical symplectic algebra}
The generators of the graphical symplectic algebra ($\GSACategory$) are given by the following diagrams, where $a,b\in\reals$.
The interpretation of these diagrams are given as affine Lagrangian relations:
\begin{align*}
    \raisebox{-0.75em}{\includegraphics{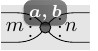}}
    \qquad & \overset{\interp{\cdot}}{\longmapsto} \qquad
    \left\{ \left(
    \begin{bmatrix}
            \vec z \\ x \vspace*{-.2cm}\\ \vdotsm \\ x
        \end{bmatrix} ,
    \begin{bmatrix}
            \vec {z'} \\ x \vspace*{-.2cm} \\ \vdotsn \\ x
        \end{bmatrix}
    \right) \,\middle|\, \begin{aligned} &\vec z\in\reals^m, \vec{z'} ; \in \reals^n, x \in \reals \qq{such that} \\ &\sum_{j=0}^{m-1} z_j - \sum_{k=0}^{n-1} z_k' + bx = a \end{aligned} \right\} \\
    \raisebox{-0.75em}{\includegraphics{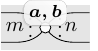}}
    \qquad & \overset{\interp{\cdot}}{\longmapsto} \qquad
    \left\{ \left(
    \begin{bmatrix}
            z \vspace*{-.2cm} \\ \vdotsm \\ z \\  \vec x
        \end{bmatrix} ,
    \begin{bmatrix}
            -z \vspace*{-.2cm} \\ \vdotsn \\ -z \\  \vec {x'}
        \end{bmatrix}
    \right) \,\middle|\, \begin{aligned} & \vec x \in \reals^m, \vec{x'}\in \reals^n, z\in \reals \qq{such that} \\ &\sum_{j=0}^{m-1} x_j + \sum_{k=0}^{n-1} x_k' - bz = a \end{aligned} \right\} \\
    \raisebox{-0.75em}{\includegraphics{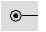}}
    \qquad & \overset{\interp{\cdot}}{\longmapsto} \qquad
    \left\{ \left(\bullet,
    \begin{bmatrix}
            i x \\  x
        \end{bmatrix}
    \right) \,\middle|\, x \in \reals \right\}
\end{align*}
The rules of GSA are given in Figure~\ref{fig:GSA-Rules}.
\begin{figure}[ht]
    \begin{mdframed}
        \includegraphics{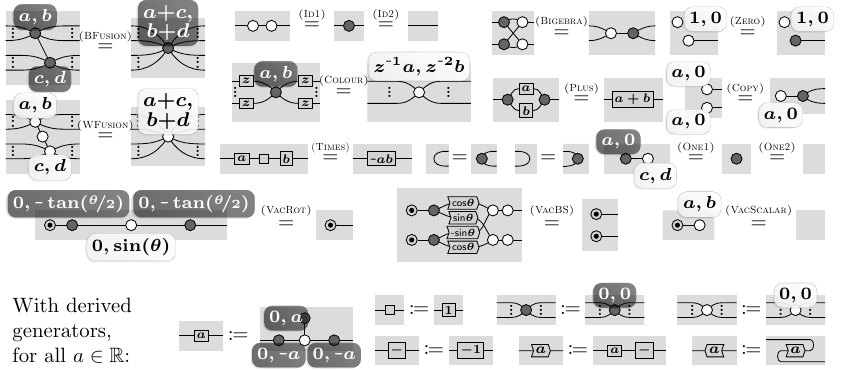}
    \end{mdframed}
    \caption{Rules of the graphical symplectic algebra (GSA).}\label{fig:GSA-Rules}
\end{figure}

\subsection{Gaussian fragment of the ZX-calculus}
We define the Gaussian fragment of the ZX-calculus, denoted by $\ZXGaussCategory$, to consist of diagrams which can be obtained from the following generators:
\begin{equation*}
    \tikzfigscale{1}{gaussian-fragment}
\end{equation*}
where $a,b\in\reals$, modulo the equations in Figure~\ref{fig:ZX-Rules}.
Additionally, since the Gaussian completeness of~\cite{boothCompleteGaussian2024} is only up to non-zero global scalars, we will ignore the global scalars for the purposes of this section, by (temporarily) adding the following axiom:
\begin{equation}\label{eq:ignore-scalars}
    \tikzfigscale{1}{ignore-scalars} \qquad \text{for all non-zero } k.
\end{equation}
We proceed to define a translation functor $T: \ZXGaussCategory \to \GSACategory$.
Since both $\ZXGaussCategory$ and $\GSACategory$ are props freely generated by a set of generators, it suffices to define the functor on the generators:
\begin{align}
    \tikzfigscale{1}{figures/T-Z} \quad      & \overset{T}{\longmapsto} \quad \raisebox{-0.75em}{\includegraphics{figures/GSA-Z.pdf}}       \\
    \tikzfigscale{1}{figures/T-X} \quad      & \overset{T}{\longmapsto} \quad \raisebox{-0.75em}{\includegraphics{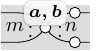}} \\
    \tikzfigscale{1}{figures/T-vacuum} \quad & \overset{T}{\longmapsto} \quad \raisebox{-0.4em}{\includegraphics{figures/GSA-vacuum.pdf}}
\end{align}
where $a,b\in\reals$.
It is straightforward to see that the translation functor $T$ is invertible.
\begin{lemma} \label{lem:invertible}
    For any diagram $D$ in $\ZXGaussCategory$, $T^{-1}(T(D)) = D$.
\end{lemma}

\noindent The semantics of $\GSACategory$ is given in terms of affine Lagrangian relations, whereas the semantics of $\ZXGaussCategory$ is given in terms of linear maps in $\starHilbCategory$, the category of hyperfinite-dimensional non-standard Hilbert spaces (cf. Appendix \ref{sec:non-standard}).
To lift the completeness of $\GSACategory$ to $\ZXGaussCategory$, we need to show the relationship between the two semantic categories.
In the next proposition, we show that $\AffLagRel_{\complexs}^+$ is equivalent to a subcategory of $\starHilbCategory$.
We start by restricting to the Gaussian fragment of $\starHilbCategory$ and quotienting by non-zero scalars.
We call the resulting subcategory $\starHilbCategory_G$: it contains Gaussian states and operators, as well as the eigenstates of the position and momentum operators (aka infinitely squeezed states),
but it is not yet equivalent to $\AffLagRel_{\complexs}^+$, as it contains Gaussian operators that are not near-standard.
We then move to a subcategory $\starHilbCategory_G^{\text{fin}}$, restricted to Gaussian operators with finite parameters:
\begin{restatable}{proposition}{propSemanticsEquivalence}\label{prop:semantics-equivalence}
    There is a symmetric-monoidal equivalence $S:~\starHilbCategory_G^{\text{fin}} \to \AffLagRel_{\complexs}^+$.
\end{restatable}
\noindent Moreover, as the generators of $\ZXGaussCategory$ are parameterized by real numbers:
\begin{lemma}\label{lem:zx-gauss-semantics}
    The semantics of $\ZXGaussCategory$ lies within the subcategory $\starHilbCategory_G^{\text{fin}}$.
\end{lemma}
\noindent Next, we check that the translation functor $T$ preserves the semantics:
\begin{proposition}\label{prop:preserves-semantics}
    For any diagram $D$ in $\ZXGaussCategory$, $S(\interp{T(D)}) = \interp{D}$.
\end{proposition}
\noindent We are now ready to show that whenever two diagrams are equal in $\GSACategory$, we can therefore prove their equality using the rules of $\ZXGaussCategory$.
\begin{restatable}{proposition}{propGSARules}\label{prop:gsa-zx-rules}
    For diagrams $D_1$ and $D_2$ in $\GSACategory$, if $\GSACategory \vdash D_1 = D_2$, then $\ZXGaussCategory \vdash T^{-1}(D_1) = T^{-1}(D_2)$.
\end{restatable}
\begin{proof}
    By the functoriality of $T^{-1}$, it is sufficient to show that all the axioms of $\GSACategory$ (Figure~\ref{fig:GSA-Rules}) are derivable in $\ZXGaussCategory$.
    The table below summarizes the proofs for each rule.
    \begin{center}
        \begin{tabular}{ll}
            \toprule
            GSA rule \hspace{6em} & Follows from                                     \\
            \midrule
            \textsf{(BFusion)}    & \textruleref{Fusion}                             \\[0.2em]
            \textsf{(WFusion)}    & \textruleref{Fusion}                             \\[0.2em]
            \textsf{(Id)}         & \textruleref{Identity}                           \\[0.2em]
            \textsf{(Bialgebra)}  & \textruleref{Bialgebra}                          \\[0.2em]
            \textsf{(Zero)}       & Lemma~\ref{lem:gsa-lem-zero} \& \ref{lem:nul}    \\[0.2em]
            \textsf{(Colour)}     & \textruleref{X-Spider} \& \textruleref{Mult}     \\[0.2em]
            \textsf{(Plus)}       & \textruleref{Plus}                               \\[0.2em]
            \textsf{(Copy)}       & \textruleref{Copy}                               \\[0.2em]
            \textsf{(Times)}      & Lemma~\ref{lem:hboxFusion}                       \\[0.2em]
            \textsf{(One)}        & Lemma \ref{lem:gsa-one-1} \& \ref{lem:gsa-one-2} \\[0.2em]
            \textsf{(VacRot)}     & Lemma \ref{lem:gsa-vacuum-rot}                   \\[0.2em]
            \textsf{(VacBS)}      & Lemma \ref{lem:gsa-vacuum-bs}                    \\[0.2em]
            \textsf{(VacScalar)}  & Lemma \ref{lem:gsa-vacuum-scalar}                \\ \bottomrule
        \end{tabular}
    \end{center}
\end{proof}
\noindent We put everything together to obtain our completeness result for the Gaussian fragment:
\begin{theorem}
    $\ZXGaussCategory$ is complete for the Gaussian fragment of CVQC: For any two diagrams $D_1$ and $D_2$ in $\ZXGaussCategory$, if $\interp{D_1} = \interp{D_2}$, then $\ZXGaussCategory$ $\vdash D_1 = D_2$.
\end{theorem}
\begin{proof}
    Since $D_1$ and $D_2$ are diagrams in $\ZXGaussCategory$, by Lemma~\ref{lem:zx-gauss-semantics}, their semantics lie in $\starHilbCategory_G^{\text{fin}}$.
    Then we can apply Proposition~\ref{prop:preserves-semantics} to $\interp{D_1} = \interp{D_2}$ and obtain $S(\interp{T(D_1)}) = S(\interp{T(D_2)})$.
    By Proposition~\ref{prop:semantics-equivalence}, we have $\interp{T(D_1)} = \interp{T(D_2)}$.
    Then by the completeness of $\GSACategory$, we have $\GSACategory$ $\vdash T(D_1) = T(D_2)$.
    Next, by Proposition~\ref{prop:gsa-zx-rules}, we have $\ZXGaussCategory$ $\vdash T^{-1}(T(D_1)) = T^{-1}(T(D_2))$.
    Finally, using the invertibility of $T$ in Lemma~\ref{lem:invertible}, we have $\ZXGaussCategory$ $\vdash D_1 = D_2$.
\end{proof}

\section{Quantum Error Correction with the GKP code}
\label{sec:gkp}
Quantum systems are delicate, because quantum information is prone to errors. If quantum computers fail to keep up with correction as errors occur, then the accumulating backlog causes exponential slowdown of the computation~\cite{terhal2015qecqmems}.
By the quantum threshold theorem, if the physical error rate of a quantum computation is below a particular threshold, then the logical error rate can be made arbitrarily small by recursively concatenating quantum error correcting codes.

In addition to physical noise, errors in the continuous-variable regime can also be due to the practical limitations of squeezing: infinitely squeezed Gaussian CV states are unphysical, because they would require infinite energy to prepare.
Physical implementations approximate them by finitely squeezed Gaussian states instead, their wavefunctions being peaks with finite non-vanishing width.

A common approach to CV quantum error correction is to concatenate an inner bosonic code --- such as GKP~\cite{GottesmanKitaevPreskill2001}, cat~\cite{CochraneCat1999, RalphCat2003}, or binomial~\cite{MichaelBinomial2016} codes --- with an outer qubit code.
In bosonic codes, qubits are encoded into squeezed Gaussian CV states, such that fault-tolerant quantum computation is achievable for some finite squeezing threshold.
Squeezing above this threshold makes it theoretically possible to do arbitrarily long fault-tolerant qubit computations, encoding the discrete logical qubits into the continuous degrees of freedom of the physical bosonic hardware~\cite{Menicucci2014ftmbqccv}.
Experimental demonstrations of GKP codes have surpassed the break-even point where computations which use quantum error correction outperform computations which don't~\cite{sivak2023gkpbreakeven}, but are not yet robust enough for full-fledged fault-tolerance.

\subsection{Representing the GKP code}
\subsubsection{States}
The GKP code encodes one logical qubit per CV mode.
The position and momentum bases have two computational basis states each --- $\{\ket{0_L}, \ket{1_L}\}$ and $\{\ket{+_L}, \ket{-_L}\}$.
In the simplest and ideal case of the GKP code, these four states are encoded as equal superpositions of infinitely squeezed Gaussian states, equally spaced at $2\sqrt{\pi}$ intervals.
The `comb' of $\ket{0_L}$ is offset $\sqrt{\pi}$ from the `comb' of $\ket{1_L}$ in the position basis, as is $\ket{+_L}$ from $\ket{-_L}$ in the momentum basis:
\begin{align*}
    \ket{0_L} & \defeq \sum_{k \in \integers} \ket{2 k \sqrt{\pi}}_X       & \ket{+_L} & \defeq \sum_{k \in \integers} \ket{2 k \sqrt{\pi}}_P       \\
    \ket{1_L} & \defeq \sum_{k \in \integers} \ket{(2 k + 1) \sqrt{\pi}}_X & \ket{-_L} & \defeq \sum_{k \in \integers} \ket{(2 k + 1) \sqrt{\pi}}_P
\end{align*}
Since the ZX calculus allows us to label the spider by any function $\reals\to\complexs$, we can represent these states directly as:
\begin{gather*}
    \tikzfigscale{1}{zeroL} \qquad\qquad\qquad\qquad \tikzfigscale{1}{plusL}\\
    \tikzfigscale{1}{oneL} \qquad\qquad\qquad\qquad \tikzfigscale{1}{minusL}
\end{gather*}
where $0_L(x) = \sum_{k \in \integers} \delta(x - 2 k \sqrt{\pi})$ and $1_L(x) = \sum_{k \in \integers} \delta(x - (2 k + 1) \sqrt{\pi})$.
The above ideal GKP states are unphysical and must be approximated by finitely squeezed Gaussian states.
These physical GKP states approximate each Dirac delta with a Gaussian of width $\Delta$, the sum of these Gaussian itself weighted by an overall Gaussian envelope of width $\frac{1}{\Delta}$:
\begin{align}
    \ket{\tilde{0}_L} & = \int \sum_{k \in \integers} e^{-2 \pi \Delta^2 k^2} e^{-\frac{x^2}{2\Delta^2}} \ket{x + 2k\sqrt{\pi}}\, dx       \\
    \ket{\tilde{1}_L} & = \int \sum_{k \in \integers} e^{-2 \pi \Delta^2 k^2} e^{-\frac{x^2}{2\Delta^2}} \ket{x + (2k + 1)\sqrt{\pi}}\, dx
\end{align}
We can use the same trick as earlier and write these states directly as:
\begin{equation}
    \tikzfigscale{1}{gkp-approx-states-direct}
\end{equation}
However, this form hides more complexity than needed inside the function labels.
Instead, we can derive a more compelling description that allows us to view the approximate GKP states as ideal GKP states with added Gaussian noise:
\begin{proposition}
    \begin{equation}
        \tikzfigscale{1}{gkp-approx-states-noise}
    \end{equation}
\end{proposition}
\begin{proof}
    We only prove for $\ket{\tilde{0}_L}$; the proof for $\ket{\tilde{1}_L}$ is analogous.
    \begin{equation*}
        \tikzfigscale{1}{gkp-approx-states-noise-proof}
    \end{equation*}
\end{proof}
\begin{proposition}
    \begin{equation}
        \tikzfigscale{1}{gkp-approx-momentum-states-noise}
    \end{equation}
\end{proposition}
\begin{proof}
    We prove for $\ket{\tilde{+}_L}$; the proof for $\ket{\tilde{-}_L}$ follows similarly.
    \begin{equation*}
        \tikzfigscale{1}{gkp-approx-momentum-states-noise-proof}
    \end{equation*}
\end{proof}


\subsubsection{GKP Encoder}
The map which transforms states in logical space to states in physical space is called the encoder.
Now we will build the encoder for the GKP code.
First note that the $\ket{0_L}$ and $\ket{1_L}$ are separated by $\sqrt{\pi}$ displacement in the position basis.
\begin{equation}
    \tikzfigscale{1}{gkp-state-shift}
\end{equation}
This allows us to define an encoder map which sends $\ket{0}_X$ to $\ket{0_L}$ and $\ket{1}_X$ to $\ket{1_L}$.
\begin{equation}
    \tikzfigscale{1}{gkp-encoder}
\end{equation}
If we use the notation from the mixed-dimensional ZX-calculus introduced in~\cite{wangCompletenessQufiniteZXW2024,poorZXcalculusCompleteFiniteDimensional2024}, we can write the encoder map for the GKP code exactly as:
\begin{equation}
    \tikzfigscale{1}{gkp-encoder-mixed}
\end{equation}
Here, the mixed-dimensional Z spider is simply an embedding of the standard basis; i.e.\@ it sends qubit $\ket{0}$ to $\ket{0}_X$ and $\ket{1}$ to $\ket{1}_X$.
Similarly, for the approximate GKP states, we have the following mixed-dimensional encoder:
\begin{equation}
    \tikzfigscale{1}{gkp-encoder-approx-mixed}
\end{equation}

\subsection{Syndrome measurements and correction}
Errors for continuous states are themselves continuous, posing an additional complication for CVQC error correction when compared to discrete error correction.
Fortunately, in the GKP encoding sufficiently small continuous errors can be corrected, by projecting displacement errors in position and momentum space back to the discrete setting where GKP states live.
A round of quantum error correction in the GKP code consists of two syndrome measurements, followed by a correction conditioned on these measurement outcomes. The circuit whose measurement outcome is a squeezed state in the position basis $\bra{s}_X$ is shown in action below:
\begin{equation}
    \tikzfigscale{1}{syndrome-meas}
\end{equation}
In other words, the action of this syndrome measurement is to apply the position basis projection $P = \mathcal{F}(0_L) = \sum_{k \in \mathbb{Z}} \ket{k \sqrt{\pi}}_X\bra{k \sqrt{\pi}}_X$, conjugated by a position displacement operator $D(s)$. This projection is precisely that which annihilates in the position basis all single-mode CVQC states outside the the code space, while preserving all information in the code space.
To end the round, we simply have to apply a correction of $-(s \mod \sqrt{\pi})$, conditioned on the classical measurement outcome $s$.

Consider inputting an arbitrary error-free single-mode GKP-encoded state $\tikzfigscale{1}{gkp-state}$, for some $\alpha,\beta \in \mathbb{C}$ such that $\abs{\alpha}^2 + \abs{\beta}^2 = 1$.
By the above graphical reasoning, the only physically observable measurement outcomes must be of the form $s \in \sqrt{\pi}\mathbb{Z}$: otherwise, it would have been annihilated by $P$ and hence had zero probability of being observed.

Then, consider an arbitrary GKP state with an arbitrary position displacement error of $\epsilon$.
Inputting this erroneous state $\tikzfigscale{1}{gkp-state-x-disp}$ to our above syndrome measurement, we find that
\begin{equation}
    \tikzfigscale{1}{gkp-x-disp-syndrome}
\end{equation}
As both the error-free GKP state and the projector $P$ are zeroed at all positions which are not integer multiples of $\sqrt{\pi}$, the displacement $D(\epsilon-s)$ between them must be an integer multiple of $\sqrt{\pi}$ for the measurement outcome $s$ to have a nonzero probability of being observed.
Therefore, for all $\epsilon$ in $\left(-\frac{\sqrt{\pi}}{2}, \frac{\sqrt{\pi}}{2}\right)$, this syndrome measurement circuit detects the error by measuring $s = \epsilon$.
The error must be small enough for this work: if $\epsilon$ falls outside the correctable range, then instead of shifting to correct the error, the correctional shift of $-(s \mod \sqrt{\pi})$ induces a logical bit flip error, by collapsing the shifted version of the $\ket{0_L}$ comb onto the $\ket{1_L}$ comb, and vice versa.

Finally, consider a GKP state with arbitrary displacement error, which can always be described as $\tikzfigscale{1}{gkp-state-p-disp}$.
To correct this, two syndrome measurements are needed: one each in position and momentum bases.
These along with their corrections comprise a round of quantum error correction in the GKP code.
Contracting the circuit, we arrive at a fully graphical specification of the error correction procedure:
\begin{equation*}
    \tikzfigscale{1}{gkp-round}
\end{equation*}

\subsection{Magic State Distillation of the Hadamard eigenstates}
In this section, we look at a protocol for distilling the Hadamard eigenstates, which can be then used to implement the qubit T gate in the GKP code~\cite{GottesmanKitaevPreskill2001}.
The Hadamard eigenstates can be distilled by the following procedure, where one of two entangled qumodes is destructively measured in the Fock basis, observing a measurement outcome $n \in \mathbb{N}$:
\begin{equation} \label{eq:msdh}
    \tikzfigscale{1}{msdh}
\end{equation}
The qubit Hadamard gate in the GKP code is physically implementable by the Fourier gate.
By noting that $\tikzfigscale{1}{gkp-projector}$ is the projector annihilating a state supported only in the GKP codespace, we can show that the state obtained by measuring the above circuit is indeed the Hadamard eigenstate whenever $n$ is even.
\begin{proposition}
    The state obtained in \eqref{eq:msdh} is a GKP eigenstate of the Hadamard gate when the measurement outcome $n$ is even.
\end{proposition}
\begin{proof}
    Since the GKP Hadamard gate is implemented by the Fourier gate, we can plug \eqref{eq:msdh} into the Fourier gate and simplify to show that it is an eigenstate:
    \begin{equation*}
        \tikzfigscale{1}{hadamard-eigenstate-proof}
    \end{equation*}
    where in the last step we have used the fact that the projectors commute necessarily for the code to be well-defined.
    When $n=0 \mod 4$, we obtain $i^n = 1$ and it corresponds to the $+1$ Hadamard eigenstate.
    Whereas $n=2 \mod 4$ gives $i^n = -1$, the $-1$ Hadamard eigenstate.
\end{proof}

\section{A simple proof of Gaussian boson sampling}
\label{sec:gbs}
\emph{Boson sampling} is a non-universal model of quantum computation formulated by Aaronson and Arkhipov~\cite{aaronson2011bosonsampling}, but one strongly believed to not be efficiently simulateable by classical computers.
In boson sampling, identical bosons are input to a linear interferometer, which performs a unitary transformation from the input to the output modes.  For any $N$-by-$N$ unitary matrix, this is physically implementable as a linear optical circuit, consisting of beam splitters and rotation gates~\cite{reck1994LOunitary,clements2017interferometer}. The number of photons is then measured at the outputs, thereby sampling from the probability distribution associated to that interferometer.

In \emph{Gaussian boson sampling}~\cite{Hamilton2017GBS,Kruse2019gbs2}, the inputs to the interferometer are instead squeezed vacuum states; in practice, these are preferred over number states because they are easier to prepare experimentally.
Like boson sampling, Gaussian boson sampling is a classically hard-to-simulate process which can be performed efficiently on continuous-variable quantum computers.
Even though Gaussian states and processes are themselves efficiently classically simulatable, the number basis measurements used to sample the final distribution are non-Gaussian, and make Gaussian boson sampling a \#P-hard problem.
In addition to practical applications---such as to finding dense subgraphs~\cite{arrazola2018gbsdensesubgraphs}, molecular spectroscopy~\cite{huh2017gbsvibmolspec} and pharmaceuticals~\cite{banchi2020gbsmoleculardocking}---experimental realizations of Gaussian boson sampling are of direct interest as candidates for demonstrating quantum advantage~\cite{madsen2022qadvgbs,zhong2020gbsxpmt}.


\subsection{Matchings of Graphs}
The \emph{perfect matching} problem in graph theory asks, given a graph, to find a set of edges which covers each vertex exactly once.
For any graph $G(V,E)$, having a perfect matching necessarily implies that $\abs{V}$ is even, that the number of edges in the perfect matching is $\frac{\abs{V}}{2}$, and that the perfect matching is the smallest edge cover.
The decision problem (whether a perfect matching exists) is solvable in polynomial time, but the corresponding counting problem (the number of perfect matchings) is \#P-complete.

For bipartite graphs, the number of perfect matchings equals the permanent of the graph's biadjacency matrix.
The permanent of an $n$ by $n$ matrix $B = (b_{i,j})$ is defined as:
\begin{equation}
    \textsf{perm}(B) = \sum_{\sigma \in S_n} \prod_{i=1}^n b_{i, \sigma(i)}
\end{equation}
where $S_n$ is the symmetric group over all permutations of $1,2,...,n$.
For arbitrary graphs, the number of perfect matchings instead equals the hafnian of the graph's adjacency matrix, a generalisation of the permanent.
The hafnian of any odd sized symmetric matrix is defined to be zero, while the hafnian of a $2n$ by $2n$ symmetric matrix $A = (a_{i,j})$ is defined as follows:
\begin{equation}
    \textsf{haf}(A) = \sum_{\rho \in P_{2n}^2} \prod_{\{i,j\}\in \rho} a_{i, j}
\end{equation}
where $P_{2n}^2$ is the set of all partitions of $\{1,2,...,2n\}$ into $n$ subsets of size 2.
A few example graphs and all of their perfect matchings is shown in Figure~\ref{fig:perfect-matching-exs}.
\begin{figure}[H]
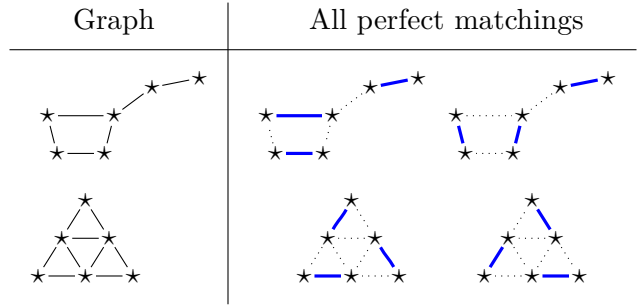

    \ctikzfig{perfect-matching-exs}
    \caption{Examples of graph perfect matchings. These are planar graphs, a special class of graphs for which counting the number of perfect matchings is solvable in polynomial time. This dipper graph (upper) is bipartite, unlike this triforce graph (lower).}\label{fig:perfect-matching-exs}
\end{figure}
Perfect matchings of a graph arise in diagrams where vertices of the graph correspond to W nodes, with number states as inputs and with outputs connected according to the graph's edges.
From the following equation:
\begin{equation}\label{eq:w-split}
    \tikzfigscale{1}{w-split}
\end{equation}
it follows that inputting the $\ket{1}$ number state to a W node, equals the sum over number states $\ket{10...0}$, $\ket{010..0}$, ... $\ket{0...01}$, as shown in the following example.
\begin{example}\label{ex:perfect-matching}
    As a toy example, we can split the below graph into a sum of 8 diagrams.
    \begin{equation}
        \tikzfigscale{1}{perfect-matching-split}
    \end{equation}
    Only 2 of the 8 terms are nonzero, so the above equals
    \begin{equation}
        \tikzfigscale{1}{perfect-matching-split2} \quad =\quad 1 + 1 \quad = \quad 2
    \end{equation}
    Thus the diagram exactly computes the number of perfect matchings of the graph.
    \begin{equation}
        \tikzfigscale{1}{perfect-matching-bialg}
    \end{equation}
    The above technique to compute the number of perfect matchings is applicable to any simple unweighted graphs.
    This is generalizable to weighted graphs by the same reasoning:
    \begin{equation}
        \tikzfigscale{1}{perfect-matching-weighted} \quad = \quad u_{\scriptscriptstyle \!1,1\!} u_{\scriptscriptstyle \!2,2\!} + u_{\scriptscriptstyle \!1,2\!} u_{\scriptscriptstyle \!2,1\!}
    \end{equation}
\end{example}
Generalizing this to arbitrary unweighted graphs, we can compute the number of perfect matchings as a special case of:
\begin{proposition} \label{prop:perfect-matching}
    For a weighted adjacency matrix $A$ of a graph with $s$ vertices,
    \begin{equation}
        \tikzfigscale{1}{perfect-matching-general} \quad = \quad \textsf{haf}(A)
    \end{equation}
\end{proposition}
In the remainder of this section, we further generalize the above to graphically prove that hafnians are computed by Gaussian boson sampling.

\subsection{Normal form for interferometers}
Interferometers admit an efficient algorithm to a graphical normal form~\cite{bonchi2014cattheosignalflowgraphs,defeliceQuantumLinearOptics2022}:
\begin{proof}[Algorithm for interferometer normal form]
    \phantom{\,}\\[-1em]
    \begin{enumerate}
        \item For any internal (connected to no inputs or outputs) W node, rewrite per Figure~\ref{fig:boson-sampling-steps}.
        \item Repeat Step 1\@ until all left-pointing W nodes are connected to inputs and all right-pointing W nodes are connected to outputs.
        \item Repeat Steps 1 and 2 until no internal W nodes remain.
        \item Fuse all sequential Fock spiders which can be fused.
        \item Fuse all parallel Fock spiders which can be fused, by Corollary~\ref{cor:f-plus-gen}.
        \item For inputs or outputs not connected to a W node, a 1-input 1-output W node (which equals the identity) can be inserted. Likewise, for internal wires not connected to any Fock spider, a 1-input 1-output $1^{\hat{n}}$ identity Fock spider can be inserted.
    \end{enumerate}
    \begin{figure}[H]
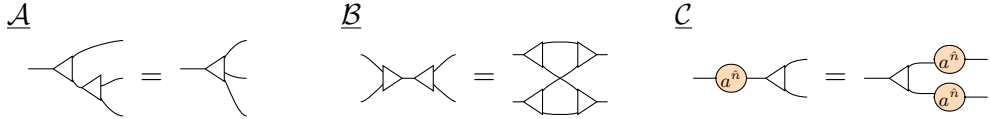

        \ctikzfig{boson-sampling-steps}
        \caption{Consider the W node left of the equals sign. No matter which node it points to, the corresponding rewrite always reduces its depth (number of nodes between it and the closest input) by 1. Transpose these rules to rewrite right-pointing W nodes closer to outputs. Note that, although rewrite $\mathcal{B}$ creates new W nodes, these are always of strictly lower depth than the two internal W nodes being rewritten.}\label{fig:boson-sampling-steps}
    \end{figure}
\end{proof}

\noindent For an interferometer on $n$ modes, the resulting normal form has three layers of nodes. The left layer is $n$ left-pointing W nodes, and the right layer is $n$ right-pointing W nodes. Between them is a bipartite graph, with a 1-input 1-output Fock spider on each edge.
This algorithm terminates in time polynomial in the number of nodes of the initial diagram, which is linear in the number of beamsplitters and rotation gates. This can be deduced by considering all internal W nodes at a given time, and applying the appropriate rewrite in Figure~\ref{fig:boson-sampling-steps}. This always reduces the maximum depth of all internal W nodes by 1.

\subsection{Normal form for Gaussian boson sampling}
In the previous section, we reduced the interferometer to a normal form: a bipartite graph of W nodes, where the edge between $i$-th input and $j$-th output is labelled by a Fock spider $(u_{ij})^{\hat{n}}$, where $u_{ij}$ are elements of the unitary matrix $U$ identified with the interferometer.
The input state is the squeezed vacuum state, which we can represent by starting from a vacuum state and applying a multiplier labelled by the squeezing parameter $r$.
Putting this together, we get the following diagram for the Gaussian boson sampling circuit, where squeezed vacuum states are input into an interferometer in normal form.
\begin{equation}\label{eq:gbs-circuit}
    \tikzfigscale{1}{gbs-circuit}
\end{equation}
To rewrite this further, we first derive few lemmas below. Lemmas in this section are proved in Appendix~\ref{sec:gbs-appendix}.
Throughout this section, for some weighted adjacency matrix $A$ of size $s$, we write the following to mean all-to-all connectivity of the W-nodes, where an edge between $i$-th and $j$-th W node has a Fock spider labelled by $a_{ij}^{\hat{n}}$:
\begin{equation}
    \tikzfigscale{1}{complete-graph-notation}
\end{equation}
In the case where the matrix element $u_{ij}$ is zero, this has the same effect as that edge being absent:
\begin{equation}
    \tikzfigscale{1}{no-edge}
\end{equation}
This works out for a curious reason: $0^{n} = 0$ for all $n \in \mathbb{N}$, with the exception of $0^{n} = 1$ when $n = 0$, so that  $0^{\hat{n}} = \delta_0$.

\begin{restatable}{lemma}{Knloops}\label{lem:complete-graph-with-loops}
    \begin{equation}
        \tikzfigscale{1}{complete-graph-with-loops}
    \end{equation}
    where $J_s$ is an $s\times s$ matrix of all ones.
\end{restatable}
\begin{restatable}{lemma}{Knouter}\label{lem:complete-graph-outer-product}
    For some $a$ and $B = [b_1, \dots, b_s]^T$, we have
    \begin{equation}
        \tikzfigscale{1}{complete-graph-outer-product}
    \end{equation}
    where $J_s$ is an $s\times s$ matrix of all ones.
\end{restatable}
\begin{restatable}{lemma}{matrixadd}\label{lem:matrix-addition}
    \begin{equation}
        \tikzfigscale{1}{matrix-addition}
    \end{equation}
\end{restatable}
This enables further rewriting of the Gaussian boson sampling circuit of Equation~\eqref{eq:gbs-circuit}.
\begin{theorem}\label{thm:gbs-rewrite-normal-form}
    The circuit of Gaussian boson sampling can be reduced to the following normal form:
    \begin{equation}
        \tikzfigscale{1}{gbs-rewrite-normal-form}
    \end{equation}
    where $U$ is the matrix of the interferometer, $r_i$ represents the amount of squeezing on the $i$-th mode, and $T_i = \frac{-\textsf{tanh}\,r_i}{2} \sum\limits_{j = 1}^s  u_{ij}^2$.
\end{theorem}
\begin{proof}
    We start with the form of the Gaussian boson sampling circuit in Equation~\eqref{eq:gbs-circuit}.
    In the diagrams below, all occurrences of ``$...$'' denote exactly $s$ wires, for $s$ the number of modes of the Gaussian boson sampling circuit.
    \begin{equation*}
        \tikzfigscale{1}{gbs-rewrite-normal-form-proof-0}
    \end{equation*}
    In the following steps, we will omit the scalar factor for brevity.
    \begin{align*}
        \tikzfigscale{1}{gbs-rewrite-normal-form-proof}\\[2em]
        \tikzfigscale{1}{gbs-rewrite-normal-form-proof-2}
    \end{align*}
    where $T_i = \frac{-\textsf{tanh}\,r_i}{2} \sum\limits_{j = 1}^s u_{ij}^2$ for $1 \leq t \leq s$.
\end{proof}

\subsection{Measuring 0-1 photons per mode}
In Gaussian boson sampling, after the interferometer, the photon counts are observed.
First, we consider the case where these photon counts are all either 0 or 1.
Whenever the measurement outcome is the number state $\ket{0}$, the self-loop can be ignored due to Lemma~\ref{lem:vacuumCopy}.
Whenever the measurement outcome is $\ket{1}$, the self-loop can also be ignored:
\begin{restatable}{lemma}{loopPop}\label{lem:loop-pop}
    \begin{equation}
        \tikzfigscale{1}{loop-pop}
    \end{equation}
    for any nonzero $k \in \mathbb{C}$.
\end{restatable}
\noindent Because of this, when we observe the measurement outcomes $n_1, \dots, n_s$ with $n_i \in \{0,1\}$, we can remove the self-loops:
\begin{equation}
    \tikzfigscale{1}{gbs-0-1-loop-removal}
\end{equation}
Now, let us consider the simplest case of observing all $\ket{0}$'s:
\begin{equation}
    \tikzfigscale{1}{gbs-all-0s}
\end{equation}
Here, all the 0s copy through the W nodes, and entire matrix is deleted.
We are simply left with the scalar factor.
Next, when all $\ket{1}$'s are measured, we get the following result:
\begin{equation}
    \tikzfigscale{1}{gbs-all-1s}
\end{equation}
where we used Proposition~\ref{prop:perfect-matching} to obtain the hafnian.
We can now combine the insights from the above examples to the case of observing arbitrary 0-1 photon counts:
\begin{proposition}
    The amplitude of observing $n_1, \dots, n_s \in \{0,1\}$ photons in the Gaussian boson sampling circuit is
    \begin{equation}
        \tikzfigscale{1}{gbs-0-1-photon-count}
    \end{equation}
    where we take the hafnian of the submatrix.
    This submatrix is obtained by removing the $i$-th row and column for each $i$ when $n_i = 0$.
\end{proposition}
\begin{proof}
    We consider the case where $n_i = 0$ for some $i$.
    \begin{equation}
        \tikzfigscale{1}{gbs-0-1-photon-count-proof}
    \end{equation}
    Hence, we have that every wire associated with the $i$-th row and column of the adjacency matrix is now connected to a vacuum state $\delta_0$.
    Then we use the fact that
    \begin{equation}
        \tikzfigscale{1}{gbs-0-1-photon-count-proof2}
    \end{equation}
    to obtain the result.
\end{proof}

\subsection{Measuring arbitrary photon counts}
The above results can be generalized to the case where we observe multiple photons in each mode.
\textcite{Kruse2019gbs2} reduced the problem of computing the probability of multiple photon counts to the case of 0-1 photon counts by ``artificially moving each photon to another pseudo-mode''.
This insight corresponds to the following lemma in the ZX calculus.
\begin{restatable}{lemma}{fdeltam}\label{lem:f-delta-m}
    \begin{equation}
        \tikzfigscale{1}{f-delta-m}
    \end{equation}
    for any $m \in \mathbb{N}$, where $J_m$ is an $m \times m$ matrix of all ones.
\end{restatable}
The amplitude is then proportional to the hafnian of a new matrix, where we repeat the $i$-th row and column $n_i$ times.
We use the above to show how the rows and columns of the adjacency matrix are copied.
We first deal with the diagonal elements, i.e.\@ the elements corresponding to the self-loops of the W nodes.
For this, we have the following lemma:
\begin{restatable}{lemma}{loopPopGen}\label{lem:loop-pop-gen}
    \begin{equation}
        \tikzfigscale{1}{loop-pop-gen}
    \end{equation}
    for any $m \in \mathbb{N}$.
\end{restatable}
\noindent Next, for the adjacency matrix without self-loops, i.e.\@ diagonal is zero, the rows and columns are copied as follows:
\begin{restatable}{lemma}{matrixCopy}\label{lem:matrix-copy}
    For an adjacency matrix $A$ with zero diagonal, if we measure $\delta_m$ on the $i$-th wire, we obtain
    \begin{equation}
        \tikzfigscale{1}{gbs-matrix-copy}
    \end{equation}
    where $A'$ is the matrix obtained by copying the $i$-th row and column $m$ times.
\end{restatable}
\begin{proof}
    We first apply \textruleref{Bialgebra} on the $i$-th wire.
    \begin{equation*}
        \tikzfigscale{1}{gbs-matrix-copy-proof}
    \end{equation*}
    Then we have Fock spiders labelled $a_{ij}^{\hat{n}}$ for each $j$ connected to the W-nodes of the bialgebra.
    This corresponds to the $i$-th row of the matrix.
    We can use the Lemma~\ref{lem:push} to push the Fock spiders, which copies each element of the $i$-th row $m$ times.
    \begin{equation*}
        \tikzfigscale{1}{gbs-matrix-copy-proof2}
    \end{equation*}
    Finally since we do not have any self-loops, all the internal W-nodes that arise from the bialgebra can be fused, and we obtain the result.
\end{proof}
\noindent With this, we have all the ingredients to derive the amplitude of Gaussian boson sampling for arbitrary photon counts:
\begin{theorem}
    The amplitude of observing $n_1, \dots, n_s \in \naturals$ photons in the Gaussian boson sampling circuit is
    \begin{equation*}
        \tikzfigscale{1}{gbs-arbitrary-photon-count}
    \end{equation*}
    where the submatrix is obtained by copying the $i$-th row and column $n_i$ times for each $i$.
\end{theorem}
\begin{proof}
    We will omit the scalar factor for conciseness.
    For brevity, let $A = U \bigoplus\limits_{i=1}^s \textsf{tanh}\,r_i U^T$ and $A'$ be the same matrix as $A$ but with the diagonal elements zeroed.
    We begin with the normal form of the Gaussian boson sampling circuit given in Theorem~\ref{thm:gbs-rewrite-normal-form}.
    \begin{equation*}
        \tikzfigscale{1}{gbs-arbitrary-photon-count-proof}
    \end{equation*}
\end{proof}


\section*{Acknowledgements}
We thank Cole Comfort, Giovanni de Felice, Richard East and Boldizs\'ar Po\'or for insightful discussions.
RS is supported by the Clarendon Fund Scholarship.
LY is funded by the Google PhD Fellowship.

\printbibliography

@inproceedings{aaronson2011bosonsampling,
  author = {Aaronson, Scott and Arkhipov, Alex},
  title = {The computational complexity of linear optics},
  year = {2011},
  isbn = {9781450306911},
  publisher = {Association for Computing Machinery},
  doi = {10.1145/1993636.1993682},
  booktitle = {Proceedings of the Forty-Third Annual ACM Symposium on Theory of Computing},
  pages = {333–342},
  numpages = {10},
  series = {STOC '11}
}

@article{arrazola2018gbsdensesubgraphs,
  title = {Using Gaussian Boson Sampling to Find Dense Subgraphs},
  author = {Arrazola, Juan Miguel and Bromley, Thomas R.},
  journal = {Phys. Rev. Lett.},
  volume = {121},
  issue = {3},
  pages = {030503},
  numpages = {6},
  year = {2018},
  month = {07},
  publisher = {American Physical Society},
  doi = {10.1103/PhysRevLett.121.030503},
}

@article{banchi2020gbsmoleculardocking,
	title = {Molecular docking with Gaussian Boson Sampling},
	author = {Leonardo Banchi and Mark Fingerhuth and Tomas Babej and Christopher Ing and Juan Miguel Arrazola},
	doi = {10.1126/sciadv.aax1950},
	journal = {Science Advances},
	number = {23},
	volume = {6},
	year = {2020},
}

@article{bartlett_efficient_2002,
  title    = {Efficient {Classical} {Simulation} of {Continuous} {Variable} {Quantum} {Information} {Processes}},
  volume   = {88},
  issn     = {0031-9007, 1079-7114},
  doi      = {10.1103/PhysRevLett.88.097904},
  number   = {9},
  journal  = {Physical Review Letters},
  author   = {Bartlett, Stephen D. and Sanders, Barry C. and Braunstein, Samuel L. and Nemoto, Kae},
  month    = feb,
  year     = {2002},
  pages    = {097904}
}

@inproceedings{bonchi2014cattheosignalflowgraphs,
	address = {Berlin, Heidelberg},
	title = {A {Categorical} {Semantics} of {Signal} {Flow} {Graphs}},
	isbn = {978-3-662-44584-6},
	booktitle = {{CONCUR} 2014 – {Concurrency} {Theory}},
	publisher = {Springer Berlin Heidelberg},
	author = {Bonchi, Filippo and Sobociński, Paweł and Zanasi, Fabio},
	editor = {Baldan, Paolo and Gorla, Daniele},
	year = {2014},
	pages = {435--450},
}

@misc{boothCompleteGaussian2024,
  title         = {Complete equational theories for classical and quantum Gaussian relations},
  author        = {Robert I. Booth and Titouan Carette and Cole Comfort},
  year          = {2024},
  eprint        = {2403.10479},
  archiveprefix = {arXiv}
}

@misc{boothGraphicalSymplecticAlgebra2024,
  title         = {Graphical {{Symplectic Algebra}}},
  author        = {Booth, Robert I. and Carette, Titouan and Comfort, Cole},
  year          = {2024},
  month         = jan,
  eprint        = {2401.07914},
  archiveprefix = {arxiv}
}

@inproceedings{borgnaEncodingHighlevelQuantum2023,
  title     = {Encoding High-Level Quantum Programs as {{SZX-diagrams}}},
  booktitle = {Proceedings 19th International Conference on Quantum Physics and Logic, Wolfson College, Oxford, {{UK}}, 27 June - 1 July 2022},
  author    = {Borgna, Augustin and Romero, Rafael},
  year      = {2023},
  series    = {Electronic Proceedings in Theoretical Computer Science},
  volume    = {394},
  pages     = {141--169},
  publisher = {Open Publishing Association},
  doi       = {10.4204/EPTCS.394.9}
}

@article{braunsteinContinuousVariablesReview2005,
  title     = {Quantum information with continuous variables},
  volume    = {77},
  issn      = {1539-0756},
  doi       = {10.1103/revmodphys.77.513},
  number    = {2},
  journal   = {Reviews of Modern Physics},
  publisher = {American Physical Society (APS)},
  author    = {Braunstein, Samuel L. and van Loock, Peter},
  year      = {2005},
  month     = jun,
  pages     = {513–577}
}

@inproceedings{clementLovCalculus2022,
  author    = {Cl\'{e}ment, Alexandre and Heurtel, Nicolas and Mansfield, Shane and Perdrix, Simon and Valiron, Beno\^{i}t},
  title     = {{LOv-Calculus: A Graphical Language for Linear Optical Quantum Circuits}},
  booktitle = {47th International Symposium on Mathematical Foundations of Computer Science (MFCS 2022)},
  pages     = {35:1--35:16},
  year      = {2022},
  volume    = {241},
  doi       = {10.4230/LIPIcs.MFCS.2022.35}
}

@misc{clements2017interferometer,
      title={An Optimal Design for Universal Multiport Interferometers}, 
      author={William R. Clements and Peter C. Humphreys and Benjamin J. Metcalf and W. Steven Kolthammer and Ian A. Walmsley},
      year={2017},
      eprint={1603.08788},
      archivePrefix={arXiv},
      primaryClass={physics.optics}
}

@article{CochraneCat1999,
  title     = {Macroscopically distinct quantum-superposition states as a bosonic code for amplitude damping},
  volume    = {59},
  issn      = {1094-1622},
  doi       = {10.1103/physreva.59.2631},
  number    = {4},
  journal   = {Physical Review A},
  publisher = {American Physical Society (APS)},
  author    = {Cochrane, P. T. and Milburn, G. J. and Munro, W. J.},
  year      = {1999},
  month     = apr,
  pages     = {2631–2634}
}

@misc{codsiClassicallySimulatingQuantum2023,
  title         = {Classically Simulating Quantum Supremacy IQP Circuits through a Random Graph Approach},
  author        = {Codsi, Julien and van de Wetering, John},
  year          = {2023},
  month         = jan,
  eprint        = {2212.08609},
  archiveprefix = {arxiv}
}

@inproceedings{coeckeInteractingQuantumObservables2008,
  title     = {Interacting {{Quantum Observables}}},
  booktitle = {Automata, {{Languages}} and {{Programming}}},
  author    = {Coecke, Bob and Duncan, Ross},
  year      = {2008},
  pages     = {298--310},
  publisher = {Springer},
  doi       = {10.1007/978-3-540-70583-3_25},
  langid    = {english}
}

@article{coeckeThreeQubitEntanglement2011,
  title    = {Three Qubit Entanglement within Graphical {{Z}}/{{X-calculus}}},
  author   = {Coecke, Bob and Edwards, Bill},
  year     = {2011},
  month    = mar,
  journal  = {Electronic Proceedings in Theoretical Computer Science},
  volume   = {52},
  pages    = {22--33},
  issn     = {2075-2180},
  doi      = {10.4204/EPTCS.52.3}
}

@article{comfortGraphicalCalculusLagrangian2022,
  title = {A {{Graphical Calculus}} for {{Lagrangian Relations}}},
  author = {Comfort, Cole and Kissinger, Aleks},
  year = {2022},
  month = nov,
  journal = {Electronic Proceedings in Theoretical Computer Science},
  volume = {372},
  pages = {338--351},
  doi = {10.4204/EPTCS.372.24},
}

@misc{debeaudrapSimpleZXZH2023,
  title         = {Simple {{ZX}} and {{ZH}} Calculi for Arbitrary Finite Dimensions, via Discrete Integrals},
  author        = {{de Beaudrap}, Niel and East, Richard D. P.},
  year          = {2023},
  month         = apr,
  eprint        = {2304.03310},
  archiveprefix = {arxiv}
}

@inproceedings{defeliceLightMatterInteractionZXW2023,
  title         = {Light-{{Matter Interaction}} in the {{ZXW Calculus}}},
  booktitle     = {Proceedings of the Twentieth International Conference on Quantum Physics and Logic, Paris, France, 17-21st July 2023},
  author        = {{de Felice}, Giovanni and Shaikh, Razin A. and Po{\'o}r, Boldizs{\'a}r and Yeh, Lia and Wang, Quanlong and Coecke, Bob},
  year          = {2023},
  series        = {Electronic Proceedings in Theoretical Computer Science},
  volume        = {384},
  pages         = {20--46},
  publisher     = {Open Publishing Association},
  doi           = {10.4204/EPTCS.384.2},
  eprint        = {2306.02114},
  archiveprefix = {arXiv}
}

@misc{defeliceQuantumLinearOptics2022,
  title         = {Quantum {{Linear Optics}} via {{String Diagrams}}},
  author        = {{de Felice}, Giovanni and Coecke, Bob},
  year          = {2022},
  month         = oct,
  eprint        = {2204.12985},
  archiveprefix = {arxiv}
}

@misc{devismeMinimalityFiniteDimensionalZWCalculi2024,
  title         = {Minimality in {{Finite-Dimensional ZW-Calculi}}},
  author        = {{de Visme}, Marc and Vilmart, Renaud},
  year          = {2024},
  month         = jan,
  eprint        = {2401.16225},
  archiveprefix = {arxiv}
}

@article{duncanGraphtheoreticSimplificationQuantum2020,
  title     = {Graph-Theoretic {{Simplification}} of {{Quantum Circuits}} with the {{ZX-calculus}}},
  author    = {Duncan, Ross and Kissinger, Aleks and Perdrix, Simon and van de Wetering, John},
  year      = {2020},
  month     = jun,
  journal   = {Quantum},
  volume    = {4},
  pages     = {279},
  publisher = {Verein zur F{\"o}rderung des Open Access Publizierens in den Quantenwissenschaften},
  doi       = {10.22331/q-2020-06-04-279},
  langid    = {british}
}

@inproceedings{duncanRewritingMeasurementBasedQuantum2010,
  title     = {Rewriting {{Measurement-Based Quantum Computations}} with {{Generalised Flow}}},
  booktitle = {Automata, {{Languages}} and {{Programming}}},
  author    = {Duncan, Ross and Perdrix, Simon},
  year      = {2010},
  series    = {Lecture {{Notes}} in {{Computer Science}}},
  pages     = {285--296},
  publisher = {Springer},
  address   = {Berlin, Heidelberg},
  doi       = {10.1007/978-3-642-14162-1_24},
  isbn      = {978-3-642-14162-1},
  langid    = {english}
}

@inproceedings{dundar-coeckeQuantumPicturalismLearning2023,
  title         = {Quantum {{Picturalism}}: {{Learning Quantum Theory}} in {{High School}}},
  shorttitle    = {Quantum {{Picturalism}}},
  booktitle     = {2023 {{IEEE International Conference}} on {{Quantum Computing}} and {{Engineering}} ({{QCE}})},
  author        = {{D{\"u}ndar-Coecke}, Selma and Yeh, Lia and Puca, Caterina and Pfaendler, Sieglinde M.-L. and Waseem, Muhammad Hamza and Cervoni, Thomas and Kissinger, Aleks and Gogioso, Stefano and Coecke, Bob},
  year          = {2023},
  month         = sep,
  volume        = {03},
  eprint        = {2312.03653},
  pages         = {21--32},
  doi           = {10.1109/QCE57702.2023.20321},
  archiveprefix = {arxiv}
}

@article{eastAKLTStatesZXDiagramsDiagrammatic2022,
  title      = {{{AKLT-States}} as {{ZX-Diagrams}}: {{Diagrammatic Reasoning}} for {{Quantum States}}},
  shorttitle = {{{AKLT-States}} as {{ZX-Diagrams}}},
  author     = {East, Richard D.P. and {van de Wetering}, John and Chancellor, Nicholas and Grushin, Adolfo G.},
  year       = {2022},
  month      = jan,
  journal    = {PRX Quantum},
  volume     = {3},
  number     = {1},
  publisher  = {American Physical Society},
  doi        = {10.1103/PRXQuantum.3.010302}
}

@article{genovese2017infinitesimal,
  author        = {Fabrizio Romano Genovese},
  title         = {The Way of the Infinitesimal},
  year          = {2017},
  url           = {https://arxiv.org/abs/1707.00459},
  archiveprefix = {arxiv},
  eprint        = {1707.00459}
}

@article{Gogioso2017InfiniteCQM,
  title     = {Infinite-dimensional Categorical Quantum Mechanics},
  volume    = {236},
  issn      = {2075-2180},
  doi       = {10.4204/eptcs.236.4},
  journal   = {Electronic Proceedings in Theoretical Computer Science},
  publisher = {Open Publishing Association},
  author    = {Gogioso, Stefano and Genovese, Fabrizio},
  year      = {2017},
  month     = {01},
  pages     = {51–69}
}

@phdthesis{Gogioso2017Thesis,
  title         = {Categorical Quantum Dynamics},
  author        = {Gogioso, Stefano},
  year          = {2017},
  eprint        = {1709.09772},
  archiveprefix = {arXiv},
  school        = {University of Oxford}
}

@article{Gogioso2018TowardsQFT,
  title     = {Towards Quantum Field Theory in Categorical Quantum Mechanics},
  volume    = {266},
  issn      = {2075-2180},
  doi       = {10.4204/eptcs.266.22},
  journal   = {Electronic Proceedings in Theoretical Computer Science},
  publisher = {Open Publishing Association},
  author    = {Gogioso, Stefano and Genovese, Fabrizio},
  year      = {2018},
  month     = {02},
  pages     = {349–366}
}

@article{Gogioso2019QFTinCQM,
  title     = {Quantum Field Theory in Categorical Quantum Mechanics},
  volume    = {287},
  issn      = {2075-2180},
  doi       = {10.4204/eptcs.287.9},
  journal   = {Electronic Proceedings in Theoretical Computer Science},
  publisher = {Open Publishing Association},
  author    = {Gogioso, Stefano and Genovese, Fabrizio},
  year      = {2019},
  month     = {01},
  pages     = {163–177}
}

@inproceedings{Gogioso2019Dynamics,
  title        = {A Diagrammatic Approach to Quantum Dynamics},
  author       = {Gogioso, Stefano},
  booktitle    = {8th Conference on Algebra and Coalgebra in Computer Science (CALCO 2019)},
  year         = {2019},
  organization = {Schloss Dagstuhl-Leibniz-Zentrum fuer Informatik},
  doi          = {10.4230/LIPIcs.CALCO.2019.19}
}

@book{goldblatt1998hyperreals,
  author    = {Goldblatt, Robert},
  title     = {Lectures on the Hyperreals: An Introduction to Nonstandard Analysis},
  year      = {1998},
  publisher = {Springer}
}

@article{GottesmanKitaevPreskill2001,
  title     = {Encoding a qubit in an oscillator},
  volume    = {64},
  issn      = {1094-1622},
  doi       = {10.1103/physreva.64.012310},
  number    = {1},
  journal   = {Physical Review A},
  publisher = {American Physical Society (APS)},
  author    = {Gottesman, Daniel and Kitaev, Alexei and Preskill, John},
  year      = {2001},
  month     = jun
}

@article{grossHudsonTheoremFinitedimensional2006,
  title = {Hudson's Theorem for Finite-Dimensional Quantum Systems},
  author = {Gross, D.},
  date = {2006-12-01},
  journaltitle = {Journal of Mathematical Physics},
  volume = {47},
  number = {12},
  pages = {122107},
  issn = {0022-2488, 1089-7658},
  doi = {10.1063/1.2393152}
}

@phdthesis{hadzihasanovicAlgebraEntanglementGeometry2017,
  title         = {The Algebra of Entanglement and the Geometry of Composition},
  author        = {Hadzihasanovic, Amar},
  year          = {2017},
  eprint        = {1709.08086},
  archiveprefix = {arxiv},
  langid        = {english},
  school        = {University of Oxford}
}

@article{Hamilton2017GBS,
  title     = {Gaussian Boson Sampling},
  author    = {Hamilton, Craig S. and Kruse, Regina and Sansoni, Linda and Barkhofen, Sonja and Silberhorn, Christine and Jex, Igor},
  journal   = {Phys. Rev. Lett.},
  volume    = {119},
  issue     = {17},
  pages     = {170501},
  numpages  = {5},
  year      = {2017},
  month     = {10},
  publisher = {American Physical Society},
  doi       = {10.1103/PhysRevLett.119.170501},
}

@misc{heurtel2024complete,
  title         = {A Complete Graphical Language for Linear Optical Circuits with Finite-Photon-Number Sources and Detectors},
  author        = {Nicolas Heurtel},
  year          = {2024},
  eprint        = {2402.17693},
  archiveprefix = {arXiv}
}

@article{huh2017gbsvibmolspec,
	title = {Vibronic {Boson} {Sampling}: {Generalized} {Gaussian} {Boson} {Sampling} for {Molecular} {Vibronic} {Spectra} at {Finite} {Temperature}},
	volume = {7},
	copyright = {2017 The Author(s)},
	issn = {2045-2322},
	shorttitle = {Vibronic {Boson} {Sampling}},
	doi = {10.1038/s41598-017-07770-z},
	language = {en},
	number = {1},
	journal = {Scientific Reports},
	author = {Huh, Joonsuk and Yung, Man-Hong},
	month = aug,
	year = {2017},
	note = {Publisher: Nature Publishing Group},
	pages = {7462},
}

@inproceedings{huangGraphicalCSSCode2023,
  title     = {Graphical {{CSS}} Code Transformation Using {{ZX}} Calculus},
  booktitle = {Proceedings of the Twentieth International Conference on Quantum Physics and Logic},
  author    = {Huang, Jiaxin and Li, Sarah Meng and Yeh, Lia and Kissinger, Aleks and Mosca, Michele and Vasmer, Michael},
  year      = {2023},
  series    = {Electronic Proceedings in Theoretical Computer Science},
  volume    = {384},
  pages     = {1--19},
  publisher = {Open Publishing Association},
  doi       = {10.4204/EPTCS.384.1}
}

@misc{jeandel2018zxrationalangle,
      title={A Generic Normal Form for ZX-Diagrams and Application to the Rational Angle Completeness}, 
      author={Emmanuel Jeandel and Simon Perdrix and Renaud Vilmart},
      year={2018},
      eprint={1805.05296},
      archivePrefix={arXiv},
      primaryClass={quant-ph}
}

@inproceedings{jeandelCompleteAxiomatisationZXCalculus2018,
  title         = {A {{Complete Axiomatisation}} of the {{ZX-Calculus}} for {{Clifford}}+{{T Quantum Mechanics}}},
  booktitle     = {Proceedings of the 33rd {{Annual ACM}}/{{IEEE Symposium}} on {{Logic}} in {{Computer Science}}},
  author        = {Jeandel, Emmanuel and Perdrix, Simon and Vilmart, Renaud},
  year          = {2018},
  month         = jul,
  series        = {{{LICS}} '18},
  eprint        = {1705.11151},
  pages         = {559--568},
  publisher     = {Association for Computing Machinery},
  address       = {New York, NY, USA},
  doi           = {10.1145/3209108.3209131},
  archiveprefix = {arxiv},
  isbn          = {978-1-4503-5583-4}
}

@article{killoran_strawberry_2019,
  title      = {Strawberry {Fields}: {A} {Software} {Platform} for {Photonic} {Quantum} {Computing}},
  volume     = {3},
  issn       = {2521-327X},
  shorttitle = {Strawberry {Fields}},
  doi        = {10.22331/q-2019-03-11-129},
  language   = {en},
  journal    = {Quantum},
  author     = {Killoran, Nathan and Izaac, Josh and Quesada, Nicolás and Bergholm, Ville and Amy, Matthew and Weedbrook, Christian},
  month      = mar,
  year       = {2019},
  pages      = {129}
}

@misc{kolarovszki_piquasso_2024,
  title      = {Piquasso: {A} {Photonic} {Quantum} {Computer} {Simulation} {Software} {Platform}},
  shorttitle = {Piquasso},
  doi        = {10.48550/arXiv.2403.04006},
  author     = {Kolarovszki, Zoltán and Rybotycki, Tomasz and Rakyta, Péter and Kaposi, Ágoston and Poór, Boldizsár and Jóczik, Szabolcs and Nagy, Dániel T. R. and Varga, Henrik and El-Safty, Kareem H. and Morse, Gregory and Oszmaniec, Michał and Kozsik, Tamás and Zimborás, Zoltán},
  month      = mar,
  year       = {2024},
}

@article{Kruse2019gbs2,
   title={Detailed study of Gaussian boson sampling},
   volume={100},
   ISSN={2469-9934},
   DOI={10.1103/physreva.100.032326},
   number={3},
   journal={Physical Review A},
   publisher={American Physical Society (APS)},
   author={Kruse, Regina and Hamilton, Craig S. and Sansoni, Linda and Barkhofen, Sonja and Silberhorn, Christine and Jex, Igor},
   year={2019},
   month={09}
}

@article{loydContinuousVariables1999,
  title     = {Quantum Computation over Continuous Variables},
  volume    = {82},
  issn      = {1079-7114},
  doi       = {10.1103/physrevlett.82.1784},
  number    = {8},
  journal   = {Physical Review Letters},
  publisher = {American Physical Society (APS)},
  author    = {Lloyd, Seth and Braunstein, Samuel L.},
  year      = {1999},
  month     = feb,
  pages     = {1784--1787}
}

@article{madsen2022qadvgbs,
	title = {Quantum computational advantage with a programmable photonic processor},
	volume = {606},
	copyright = {2022 The Author(s)},
	doi = {10.1038/s41586-022-04725-x},
	language = {en},
	number = {7912},
	journal = {Nature},
	author = {Madsen, Lars S. and Laudenbach, Fabian and Askarani, Mohsen Falamarzi and Rortais, Fabien and Vincent, Trevor and Bulmer, Jacob F. F. and Miatto, Filippo M. and Neuhaus, Leonhard and Helt, Lukas G. and Collins, Matthew J. and Lita, Adriana E. and Gerrits, Thomas and Nam, Sae Woo and Vaidya, Varun D. and Menotti, Matteo and Dhand, Ish and Vernon, Zachary and Quesada, Nicolás and Lavoie, Jonathan},
	month = jun,
	year = {2022},
	pages = {75--81}
}

@inproceedings{mcelvanneyFlowpreservingZXcalculusRewrite2023,
  title     = {Flow-Preserving {{ZX-calculus}} Rewrite Rules for Optimisation and Obfuscation},
  booktitle = {Proceedings of the Twentieth International Conference on Quantum Physics and Logic, Paris, France, 17-21st July 2023},
  author    = {McElvanney, Tommy and Backens, Miriam},
  year      = {2023},
  series    = {Electronic Proceedings in Theoretical Computer Science},
  volume    = {384},
  pages     = {203--219},
  publisher = {Open Publishing Association},
  doi       = {10.4204/EPTCS.384.12}
}

@article{Menicucci2014ftmbqccv,
  title     = {Fault-Tolerant Measurement-Based Quantum Computing with Continuous-Variable Cluster States},
  author    = {Menicucci, Nicolas C.},
  journal   = {Phys. Rev. Lett.},
  volume    = {112},
  issue     = {12},
  pages     = {120504},
  numpages  = {5},
  year      = {2014},
  month     = {03},
  publisher = {American Physical Society},
  doi       = {10.1103/PhysRevLett.112.120504},
}

@article{MichaelBinomial2016,
  title     = {New Class of Quantum Error-Correcting Codes for a Bosonic Mode},
  volume    = {6},
  issn      = {2160-3308},
  doi       = {10.1103/physrevx.6.031006},
  number    = {3},
  journal   = {Physical Review X},
  publisher = {American Physical Society (APS)},
  author    = {Michael, Marios H. and Silveri, Matti and Brierley, R. T. and Albert, Victor V. and Salmilehto, Juha and Jiang, Liang and Girvin, S. M.},
  year      = {2016},
  month     = jul
}

@online{nagayoshiZXGraphicalCalculus2024,
  title = {{{ZX Graphical Calculus}} for {{Continuous-Variable Quantum Processes}}},
  author = {Nagayoshi, Hironari and Asavanant, Warit and Ide, Ryuhoh and Fukui, Kosuke and Sakaguchi, Atsushi and Yoshikawa, Jun-ichi and Menicucci, Nicolas C. and Furusawa, Akira},
  date = {2024-05-16},
  eprint = {2405.07246},
  eprinttype = {arxiv},
  pubstate = {preprint}
}

@misc{ngUniversalCompletionZXcalculus2017,
  title         = {A Universal Completion of the {{ZX-calculus}}},
  author        = {Ng, Kang Feng and Wang, Quanlong},
  year          = {2017},
  month         = jun,
  eprint        = {1706.09877},
  archiveprefix = {arxiv}
}

@book{nielsen1918hermite,
	series = {Mathematisk-fysiske meddelelser},
	title = {Recherches sur les polynomes d'{Hermite}},
	url = {https://books.google.co.uk/books?id=vkcozwEACAAJ},
	publisher = {Kongelige Danske videnskabernes selskab},
	author = {Nielsen, N. and selskab, Kongelige Danske videnskabernes},
	year = {1918},
}

@inproceedings{poorCompletenessArbitraryFinite2023,
  title         = {Completeness for Arbitrary Finite Dimensions of {{ZXW-calculus}}, a Unifying Calculus},
  booktitle     = {2023 38th {{Annual ACM}}/{{IEEE Symposium}} on {{Logic}} in {{Computer Science}} ({{LICS}})},
  author        = {Po{\'o}r, Boldizs{\'a}r and Wang, Quanlong and Shaikh, Razin A. and Yeh, Lia and Yeung, Richie and Coecke, Bob},
  year          = {2023},
  month         = jun,
  eprint        = {2302.12135},
  pages         = {1--14},
  address       = {Boston, MA, USA},
  doi           = {10.1109/LICS56636.2023.10175672},
  archiveprefix = {arxiv}
}

@inproceedings{poorQupitStabiliserZXtravaganza2023,
  title      = {The {{Qupit Stabiliser ZX-travaganza}}: {{Simplified Axioms}}, {{Normal Forms}} and {{Graph-Theoretic Simplification}}},
  shorttitle = {The {{Qupit Stabiliser ZX-travaganza}}},
  booktitle  = {Proceedings of the Twentieth International Conference on Quantum Physics and Logic, Paris, France, 17-21st July 2023},
  author     = {Po{\'o}r, Boldizs{\'a}r and Booth, Robert I. and Carette, Titouan and {van de Wetering}, John and Yeh, Lia},
  year       = {2023},
  series     = {Electronic Proceedings in Theoretical Computer Science},
  volume     = {384},
  pages      = {220--264},
  publisher  = {Open Publishing Association},
  doi        = {10.4204/EPTCS.384.13}
}

@misc{poorZXcalculusCompleteFiniteDimensional2024,
  title         = {{{ZX-calculus}} Is {{Complete}} for {{Finite-Dimensional Hilbert Spaces}}},
  author        = {Po{\'o}r, Boldizs{\'a}r and Shaikh, Razin A. and Wang, Quanlong},
  year          = {2024},
  month         = may,
  number        = {arXiv:2405.10896},
  eprint        = {2405.10896},
  archiveprefix = {arxiv}
}

@article{RalphCat2003,
  title     = {Quantum computation with optical coherent states},
  volume    = {68},
  issn      = {1094-1622},
  doi       = {10.1103/physreva.68.042319},
  number    = {4},
  journal   = {Physical Review A},
  publisher = {American Physical Society (APS)},
  author    = {Ralph, T. C. and Gilchrist, A. and Milburn, G. J. and Munro, W. J. and Glancy, S.},
  year      = {2003},
  month     = oct
}

@article{reck1994LOunitary,
  title={Experimental realization of any discrete unitary operator},
  author={Reck, Michael and Zeilinger, Anton and Bernstein, Herbert J and Bertani, Philip},
  journal={Physical review letters},
  volume={73},
  number={1},
  pages={58},
  year={1994},
  publisher={APS}
}

@book{robinson1974nonstandard,
  title     = {Non-standard analysis},
  author    = {Robinson, Abraham},
  year      = {1974},
  publisher = {Princeton University Press}
}

@inproceedings{royQuditZHCalculusGeneralised2023,
  title      = {The {{Qudit ZH-Calculus}}: {{Generalised Toffoli}}+{{Hadamard}} and {{Universality}}},
  shorttitle = {The {{Qudit ZH-Calculus}}},
  booktitle  = {Proceedings of the Twentieth International Conference on Quantum Physics and Logic, Paris, France, 17-21st July 2023},
  author     = {Roy, Patrick and {van de Wetering}, John and Yeh, Lia},
  year       = {2023},
  series     = {Electronic Proceedings in Theoretical Computer Science},
  volume     = {384},
  pages      = {142--170},
  publisher  = {Open Publishing Association},
  doi        = {10.4204/EPTCS.384.9}
}

@misc{shaikhCategoricalSemanticsFeynman2022,
  title         = {Categorical {{Semantics}} for {{Feynman Diagrams}}},
  author        = {Shaikh, Razin A. and Gogioso, Stefano},
  year          = {2022},
  month         = may,
  eprint        = {2205.00466},
  archiveprefix = {arxiv}
}

@article{shaikhHowSumExponentiate2022,
  title         = {How to Sum and Exponentiate {{Hamiltonians}} in {{ZXW}} Calculus},
  author        = {Shaikh, Razin A. and Wang, Quanlong and Yeung, Richie},
  archiveprefix = {arxiv},
  eprint        = {2212.04462},
  volume        = {394},
  issn          = {2075-2180},
  doi           = {10.4204/eptcs.394.14},
  journal       = {Electronic Proceedings in Theoretical Computer Science},
  year          = {2023},
  month         = nov,
  pages         = {236–261}
}

@article{sivak2023gkpbreakeven,
  title   = {Real-time quantum error correction beyond break-even},
  volume  = {616},
  issn    = {1476-4687},
  doi     = {10.1038/s41586-023-05782-6},
  number  = {7955},
  journal = {Nature},
  author  = {Sivak, V. V. and Eickbusch, A. and Royer, B. and Singh, S. and Tsioutsios, I. and Ganjam, S. and Miano, A. and Brock, B. L. and Ding, A. Z. and Frunzio, L. and Girvin, S. M. and Schoelkopf, R. J. and Devoret, M. H.},
  month   = apr,
  year    = {2023},
  pages   = {50--55}
}

@misc{sutcliffeProcedurallyOptimisedZXDiagram2024,
  title = {Procedurally {{Optimised ZX-Diagram Cutting}} for {{Efficient T-Decomposition}} in {{Classical Simulation}}},
  author = {Sutcliffe, Matthew and Kissinger, Aleks},
  year = {2024},
  month = mar,
  eprint = {2403.10964},
  publisher = {arXiv},
  archiveprefix = {arxiv},
}

@article{Tasca2011cvqcspatialdof,
   title={Continuous-variable quantum computation with spatial degrees of freedom of photons},
   volume={83},
   ISSN={1094-1622},
   DOI={10.1103/physreva.83.052325},
   number={5},
   journal={Physical Review A},
   publisher={American Physical Society (APS)},
   author={Tasca, D. S. and Gomes, R. M. and Toscano, F. and Souto Ribeiro, P. H. and Walborn, S. P.},
   year={2011},
   month=may
}

@article{terhal2015qecqmems,
  author  = {Terhal, Barbara M.},
  doi     = {10.1103/RevModPhys.87.307},
  journal = {Reviews of Modern Physics},
  month   = apr,
  note    = {Publisher: American Physical Society},
  number  = {2},
  pages   = {307--346},
  title   = {Quantum error correction for quantum memories},
  volume  = {87},
  year    = {2015}
}

@inproceedings{townsend-teagueFloquetifyingColourCode2023,
  title     = {Floquetifying the Colour Code},
  booktitle = {Proceedings of the Twentieth International Conference on Quantum Physics and Logic},
  author    = {{Townsend-Teague}, Alex and {Magdalena de la Fuente}, Julio and Kesselring, Markus},
  editor    = {Mansfield, Shane and Val{\^i}ron, Benoit and Zamdzhiev, Vladimir},
  year      = {2023},
  series    = {Electronic Proceedings in Theoretical Computer Science},
  volume    = {384},
  pages     = {265--303},
  publisher = {Open Publishing Association},
  doi       = {10.4204/EPTCS.384.14}
}

@misc{vandeweteringOptimalCompilationParametrised2024,
  title         = {Optimal Compilation of Parametrised Quantum Circuits},
  author        = {{van de Wetering}, John and Yeung, Richie and Laakkonen, Tuomas and Kissinger, Aleks},
  year          = {2024},
  month         = jan,
  eprint        = {2401.12877},
  archiveprefix = {arxiv}
}

@misc{wangCompletenessQufiniteZXW2024,
  title         = {Completeness of Qufinite {{ZXW}} Calculus, a Graphical Language for Finite-Dimensional Quantum Theory},
  author        = {Wang, Quanlong and Po{\'o}r, Boldizs{\'a}r and Shaikh, Razin A.},
  year          = {2024},
  month         = jan,
  eprint        = {2309.13014},
  archiveprefix = {arxiv},
  copyright     = {All rights reserved}
}

@misc{wangDifferentiatingIntegratingZX2022,
  title         = {Differentiating and {{Integrating ZX Diagrams}} with {{Applications}} to {{Quantum Machine Learning}}},
  author        = {Wang, Quanlong and Yeung, Richie and Koch, Mark},
  year          = {2022},
  month         = nov,
  eprint        = {2201.13250},
  archiveprefix = {arxiv}
}

@article{zhong2020gbsxpmt,
	title = {Quantum computational advantage using photons},
	volume = {370},
	doi = {10.1126/science.abe8770},
	number = {6523},
	journal = {Science},
	author = {Zhong, Han-Sen and Wang, Hui and Deng, Yu-Hao and Chen, Ming-Cheng and Peng, Li-Chao and Luo, Yi-Han and Qin, Jian and Wu, Dian and Ding, Xing and Hu, Yi and Hu, Peng and Yang, Xiao-Yan and Zhang, Wei-Jun and Li, Hao and Li, Yuxuan and Jiang, Xiao and Gan, Lin and Yang, Guangwen and You, Lixing and Wang, Zhen and Li, Li and Liu, Nai-Le and Lu, Chao-Yang and Pan, Jian-Wei},
	month = dec,
	year = {2020},
	pages = {1460--1463},
}

\appendix

\section{Non-standard models of the CV ZX-calculus}
\label{sec:non-standard}

In this Section, we show how non-standard analysis can be used to construct a model for a fragment of our calculus, based on previous work by two of the authors \cite{Gogioso2017InfiniteCQM,Gogioso2018TowardsQFT,Gogioso2019QFTinCQM,Gogioso2019Dynamics,shaikhCategoricalSemanticsFeynman2022}.
For the full technical details of non-standard analysis and its applications to categorical modelling of quantum mechanics, we refer the reader to the seminal book by Abraham Robinson \cite{robinson1974nonstandard}, to the more modern notes by Robert Goldblatt \cite{goldblatt1998hyperreals}, to Section 3.5 in the DPhil thesis by one of the authors \cite{Gogioso2017Thesis}, or to the introductory notes on this topic by Fabrizio Genovese \cite{genovese2017infinitesimal}.
Work on providing an exact interpretation for Fock spiders and associated rules is still ongoing, but we sketch some of our initial results in Subsection \ref{appendix:subsec-fock-basis}.

\subsection{The Dagger-compact Category \texorpdfstring{$\starHilbCategory$}{*fHilb}}

We work in the category $\starHilbCategory$ of hyperfinite-dimensional non-standard complex Hilbert spaces and non-standard linear maps between them, where by hyperfinite-dimensional we mean a space isomorphic to $\starComplexs^n$ for a non-standard natural number $n \in \starNaturals$.
More precisely:
\begin{itemize}
    \item The objects of $\fHilbCategory$ take the form $\starComplexs^X$ for some hyperfinite set $X$, i.e. one with cardinality $n \in \starNaturals$. Here, $\starComplexs^X$ is the $\starComplexs$-vector space formed by non-standard functions $X \rightarrow \starComplexs$ under pointwise operations.
    \item The morphisms $\starComplexs^X \rightarrow \starComplexs^Y$ in $\fHilbCategory$ are the non-standard functions $\starComplexs^X \rightarrow \starComplexs^Y$ which are $\starComplexs$-linear.
    \item Sequential composition is function composition, and the identity morphisms are the identity functions.
    \item The category is enriched, with hyperfinite-dimensional $\starComplexs$-linear structure on all sets of morphisms $\starComplexs^X \rightarrow \starComplexs^Y$ between fixed object pairs. Allowed sums include all finite sums, as well as all infinite sums which can be lifted from standard finite-dimensional analogues by the Transfer Theorem.
\end{itemize}
The category $\starHilbCategory$ is the non-standard counterpart of the category $\fHilbCategory$ of finite dimensional complex Hilbert spaces and linear maps, and many properties of $\fHilbCategory$ extend to $\starHilbCategory$ by simple invocations of the Transfer Theorem.

For starters, we can use braket notation for states and morphisms in $\starHilbCategory$.
We can use the Kronecker delta functions in $\starComplexs^X$ to define ``kets'':
\footnote{
    Physically, states are identified with (normalised) vectors in $\starComplexs^X$ (up to global phase), and scalars are complex numbers in $\starComplexs$.
    Categorically, on the other hand, states are identified with the corresponding ``rays'', the linear functions $\starComplexs \rightarrow \starComplexs^X$, and scalars are linear functions $\starComplexs \rightarrow \starComplexs$.
    We will freely confuse between the two pictures, as is commonplace in the categorical quantum mechanics literature, and use the ket notation $|x\rangle$ for both the vector defined above and the associated ray.
}
\begin{equation}
    |x\rangle := y \mapsto \delta_{xy}
\end{equation}
Using kets and Kronecker delta functions, we can define ``bras'', as linear functionals $\starComplexs^X \rightarrow \starComplexs$:
\begin{equation}
    \langle y| := |x\rangle \mapsto \delta_{xy}
\end{equation}
We can then adopt the usual short-hands for function composition:
\begin{align}
    \langle y | x\rangle   & := \langle y | \circ | x\rangle = \delta_{xy} \in \starComplexs              \\
    |x \rangle \langle y | & := |x \rangle \circ \langle y |: \starComplexs^Y \rightarrow \starComplexs^X
\end{align}
Using the $\starComplexs$-linear structure, can express arbitrary morphisms $f: \starComplexs^X \rightarrow \starComplexs^Y$ in matrix form, for arbitrary coefficients $(f_{yx})_{x \in X, y \in Y}$:
\begin{equation}
    f = \sum_{y \in Y} \sum_{x \in X} f_{yx} |y \rangle \langle x |
\end{equation}
It is then straightforward to show that $\starHilbCategory$ is a dagger category, defining adjunction to be conjugate transposition:
\footnote{
    The complex conjugation operation is denoted by an asterisk $^*$ following the symbol to be conjugated, as in $f_{yx}^*$.
    This is not to be confused with the star $^\star$ denoting non-standard counterparts of various mathematical objects, which instead precedes the symbol to which it applies, as in $\starNaturals$ and $\starComplexs$.
}
\begin{equation}
    f^\dagger := \sum_{x \in X} \sum_{y \in Y} f_{yx}^* |x \rangle \langle y |
\end{equation}
In particular, we have $\left(|x\rangle\right)^\dagger = \langle x |$ and $\left(\langle x|\right)^\dagger = | x \rangle$ for all $x \in X$.

Analogously to its standard counterpart $\fHilbCategory$, the category $\starHilbCategory$ is dagger compact, and hence provides sound semantics for diagrammatic reasoning.
The tensor product on spaces is defined by Cartesian product of the indexing sets:
\begin{equation}
    \starComplexs^X \otimes \starComplexs^Y := \starComplexs^{X \times Y}
\end{equation}
Tensor product $f \otimes g:  \starComplexs^{X \times Y} \rightarrow \starComplexs^{Z \times W}$ of morphisms $f:  \starComplexs^{X} \rightarrow \starComplexs^{Z}$ and $g:  \starComplexs^{Y} \rightarrow \starComplexs^{W}$ is defined by Kronecker product of the corresponding matrix representations:
\begin{equation}
    f \otimes g =
    \sum_{z \in Z} \sum_{w \in W} \sum_{x \in X} \sum_{y \in Y}
    f_{zx}g_{wy}|z,w\rangle \langle x,y|
\end{equation}
where we introduced the shorthands $|z,w\rangle := |z\rangle \otimes |w\rangle$ and $\langle x,y| := \langle x| \otimes \langle y|$ for tensor products of bras and kets.
Finally, dagger compact structure with self-dual objects---the cups $\eta_{X}$ and caps $\varepsilon_{X}$ used by diagrammatic reasoning to bend wires freely---can be defined on the canonical basis using bras and kets:
\begin{align}
    \eta_{X}        & := \sum_{x \in X} |x,x\rangle: \starComplexs \rightarrow \starComplexs^{X} \otimes \starComplexs^{X}  \\
    \varepsilon_{X} & := \sum_{x \in X} \langle x,x|: \starComplexs^{X} \otimes \starComplexs^{X} \rightarrow \starComplexs
\end{align}
The snake equations follow as in the standard case (technically, by Transfer Theorem).
In particular, the category $\starHilbCategory$ is traced.
Given a morphism $f:  \starComplexs^{X \times Z} \rightarrow \starComplexs^{Y \times Z}$, the partial trace $\operatorname{Tr}_{Z}(f):  \starComplexs^{X} \rightarrow \starComplexs^{Y}$ can be defined as:
\begin{align}
    \operatorname{Tr}_{Z}(f) & :=
    \left(\operatorname{id}_Y \otimes \varepsilon_{Z} \right)
    \circ
    \left(f \otimes \operatorname{id}_{Z} \right)
    \circ
    \left(\operatorname{id}_X \otimes \eta_{Z} \right) \\
                             & =
    \sum_{y \in Y} \sum_{x \in X} \left(\sum_{z \in Z} f_{yzxz}\right) |y \rangle \langle x |
\end{align}
Its dagger compact structure makes the category $\starHilbCategory$ a sound environment for the definition of models of diagrammatic theories: diagrams can be composed in sequence and in parallel, with bends and loops.
In order for calculations to be meaningful in the standard world, however, we have to relate $\starHilbCategory$ back to the category $\HilbCategory$ of standard Hilbert spaces and bounded linear maps (or reasonable extensions thereof ordinarily used in quantum physics).
Different constructions are available for different kinds of standard Hilbert spaces, and the construction relevant to this work is discussed over the coming Subsections.

\subsection{Spaces of Square-integrable Functions}\label{sec:non-standard-lattice}

In quantum mechanics and continuous-variable quantum computing, we are interested in spaces of square-integrable functions for some measurable space.
Here, we will focus on the most common case of wavefunctions on the real line $L^2\left[\reals\right]$.
The same construction straightforwardly extends \cite{Gogioso2019Dynamics} to wavefunctions in real vector spaces $L^2\left[\reals^n\right]$, wavefunctions on discrete lattices $L^2\left[\integers^n\right]$ and wavefunctions with periodic boundary conditions $L^2\left[\prod_{j=1}^n \reals/L_j\integers\right]$, but those examples are not directly relevant to this work.

We move to the non-standard version $\starReals$ of the Abelian Lie group $\reals$ and consider some hyperfinite subset $\LatticeR \subseteq \starReals$ which approximates the standard space $\reals$ up to infinitesimals, in a sense which we now make precise.
Thanks to the group structure, we can define an equivalence relation $\simeq$ of ``infinitesimal closeness'' on the points of $\starReals$:
\begin{equation}
    x \simeq y
    \stackrel{def}{\Leftrightarrow}
    \stdpart{x-y} = 0
\end{equation}
where $\stdpartSym: \starReals \rightharpoonup \reals$ is the standard part map.
If we inject the standard space $\reals$ as a subset $\reals \subseteq \starReals$ of its non-standard counterpart,
\footnote{
    One has to be careful with the statement that $\reals \subseteq \starReals$: it is true ``externally'', in terms of the underlying set theory, but not ``internally'', in terms of statements which are provable in the non-standard model.
    The set $\reals$ itself cannot in general be expressed in the non-standard model: the mathematical object obtained by applying its usual defining axioms is its non-standard counterpart $\starReals$ instead.
    A similar statement holds for the standard part function $\stdpartSym$: it is an external function, defined in the underlying set theory, but not one which can be expressed internally to the non-standard model.
    Regardless of these model-theoretic subtleties, it is sometimes useful to be able to work ``externally'': this is, in partcular, necessary if one wishes to relate non-standard models to standard ones. However, one should take additional care not to rely on external objects and statements when deriving results about non-standard objects.
}
our requirement of infinitesimally close approximation can be formalised as:
\begin{equation}
    \forall x \in \reals.\;
    \exists x' \in \LatticeR.\;
    x \simeq x'
\end{equation}
For the three classes of Abelian Lie groups mentioned above, it is straightforward to construct infinitesimal approximations, in the form of lattices carrying a hyperfinite Abelian group structure \cite{Gogioso2019Dynamics}.
To construct our desired lattice, we fix an odd infinite non-standard natural number $\omega \in \starNaturals$, and define $\tau := \frac{\omega^2-1}{2} \in \starNaturals$, so that $2 \tau + 1 = \omega^2$.
By Transfer Theorem, we can define the hyperfinite cyclic Abelian group $\starIntegersMod{\omega^2}$ of integers modulo $\omega^2$, whose values we embed as a subset of $\starIntegers$ as follows:
\begin{equation}
    \starIntegersMod{\omega^2}
    := \left\{
    -\tau, -\tau+1, ..., -1, 0, 1, ..., \tau-1, \tau
    \right\}
\end{equation}
We can use $\starIntegersMod{\omega^2}$ to construct our desired approximating lattice $\LatticeR$ within $\starReals$:
\begin{equation}
    \LatticeR := \frac{1}{\omega} \starIntegersMod{\omega^2}
    = \suchthat{
        x \in \starReals
    }{
        x = \frac{j}{\omega},\;
        j \in \starIntegersMod{\omega^2}
    }
\end{equation}
Our hyperfinite non-standard counterpart to $\LtwoR$ is the space $\starComplexs^\LatticeR \in \operatorname{obj}\left(\starHilbCategory\right)$, for the hyperfinite lattice $\LatticeR$ defined above.
The lattice $\LatticeR$ approximates every standard real number up to at most an infinitesimal distance of $\frac{1}{\omega}$, and spans the interval from $-\tau/\omega$ to $+\tau/\omega$.
It inherits the Abelian group structure of $\starIntegersMod{\omega^2}$: this coincides with that of the surrounding space $\starReals$ for all near-standard values, i.e. at all points approximating points of $\reals$, but it wraps around at infinity, going from $\tau/\omega$ back to $-\tau/\omega$ in one step of $\frac{1}{\omega}$.
The space $\starComplexs^\LatticeR$ is therefore a group algebra, and we write it as $\LtwoRNonstd$.

\subsection{Lifting states from \texorpdfstring{$\LtwoR$}{L2[R]} to \texorpdfstring{$\LtwoRNonstd$}{*C[R]}}

Functions $\varphi \in \LtwoR$ inject into $\LtwoRNonstd$ by non-standard extension to $\starReals$ followed by restriction to the lattice:
\begin{equation}
    \varphi \mapsto \nonstd{\varphi}|_{\LatticeR}
\end{equation}
The mapping is evidently $\complexs$-linear.
To check that it is an injection, it suffices to show that it is an isometry up to infinitesimals, i.e. that $||\varphi||^2 \simeq ||\nonstd{\varphi}|_{\LatticeR}||^2$, because then:
\begin{equation}
    ||\nonstd{\varphi}|_{\LatticeR}-\nonstd{\phi}|_{\LatticeR}||^2
    \simeq 0
    \Rightarrow
    ||\varphi-\phi||^2
    = \stdpart{||\nonstd{(\varphi-\phi)}|_{\LatticeR}||^2}
    = \stdpart{||\nonstd{\varphi}|_{\LatticeR}-\nonstd{\phi}|_{\LatticeR}||^2}
    = 0
\end{equation}
From this moment onwards, we will use $\varphi$ to denote both the standard function $\varphi \in \LtwoR$ and its non-standard counterpart $\nonstd{\varphi}|_{\LatticeR} \in \LtwoRNonstd$, leaving disambiguation to context; that is, we will silently apply the injection $\LtwoR \hookrightarrow \LtwoRNonstd$ whenever we need to use a standard function in a non-standard context.

To prove that the mapping $\LtwoR \hookrightarrow \LtwoRNonstd$ is truly an injection, we observe that inner products in $\LtwoR$ are approximated, to within infinitesimal precision, by inner products in $\LtwoRNonstd$:
\begin{equation}
    \label{eq:LtwoRnInnerProdApprox}
    \int \phi(x)^* \varphi(x) d x
    \simeq
    \frac{1}{\omega}\sum_{x \in \LatticeR} \phi(x)^* \varphi(x)
\end{equation}
This is a consequence of Transfer Theorem, and it corresponds to the following limit in the standard model:
\begin{align}
    \int \phi(x)^* \varphi(x) d x
     & =
    \lim_{k \text{ odd } \rightarrow \infty}
    \frac{1}{k}\sum_{j=-\sfrac{(k^2-1)}{2}}^{\sfrac{(k^2-1)}{2}}
    \phi\left(\sfrac{j}{k}\right)^*
    \varphi\left(\sfrac{j}{k}\right)
\end{align}
As a special case of Equation \ref{eq:LtwoRnInnerProdApprox}, we conclude that our mapping from $\LtwoR$ to $\LtwoRNonstd$ is indeed an isometry up to infinitesimals, and hence an injective linear map:
\begin{equation}
    \int |\varphi(x)|^2 d x
    =
    \int \varphi(x)^* \varphi(x) d x
    \simeq
    \frac{1}{\omega}\sum_{x \in \LatticeR} \varphi(x)^* \varphi(x)
    =
    \frac{1}{\omega}\sum_{x \in \LatticeR} |\varphi(x)|^2
\end{equation}

\subsection{Lifting Tensors from \texorpdfstring{$\LtwoR$}{L2[R]} to \texorpdfstring{$\LtwoRNonstd$}{*C[R]}}

In order to do calculations, states are not enough: we also need to inject bounded linear maps $\LtwoR^{\otimes m} \rightarrow \LtwoR^{\otimes n}$ into linear maps $\LtwoRNonstd^{\otimes m} \rightarrow \LtwoRNonstd^{\otimes m}$.
We describe the construction in detail for bounded linear maps $\LtwoR \rightarrow \LtwoR$, and then generalise it to arbitrary tensors.
We start by defining approximations of the Dirac delta functions, for all positive $k \in \naturals$:
\begin{equation}
    \delta_x^{(k)}(z) := \begin{cases}
        k & \text{ if } z \in [\frac{x-\frac{1}{2}}{k}, \frac{x+\frac{1}{2}}{k}] \\
        0 & \text{ otherwise}
    \end{cases}
\end{equation}
We also define corresponding approximations of arbitrary $\phi \in \LtwoR$, for all positive $k \in \naturals$:
\begin{equation}
    \phi^{(k)}(z) :=
    \frac{1}{k}
    \sum_{j = -\sfrac{(k^2-1)}{2}}^{+\sfrac{(k^2-1)}{2}}
    \phi\left(\sfrac{j}{k}\right) \delta_{\sfrac{j}{k}}^{(k)}(z)
\end{equation}
It is easy to prove the following limit in $\LtwoR$:
\begin{equation}
    \lim_{k \rightarrow \infty} \ket{\phi^{(k)}} = \ket{\phi}
\end{equation}
By Transfer Theorem, setting $k = \omega$ yields the following in $\LtwoRNonstd$:
\begin{equation}
    \ket{\phi} \simeq \ket{\phi^{(\omega)}}
\end{equation}
On the other hand, consider a bounded linear operator $\Phi: \LtwoR \rightarrow \LtwoR$.
By continuity, we have:
\begin{equation}
    \lim_{k \rightarrow \infty} \Phi\ket{\phi^{(k)}}
    = \Phi\ket{\phi}
\end{equation}
Inspired by this, we define the following approximate coefficients:
\begin{equation}
    \Phi_{y, x}^{(k)} :=
    \bra{\delta_y^{(k)}} \Phi \ket{\delta_x^{(k)}} \in \complexs
\end{equation}
By Transfer Theorem, we conclude that, for all $\phi, \varphi \in \LtwoR$:
\begin{align}
    \bra{\varphi} \Phi \ket{\phi}
     & \simeq
    \sum_{y \in \LatticeR}
    \sum_{x \in \LatticeR}
    \varphi(y)^*\
    \Phi_{y, x}^{(\omega)}\
    \phi(x)
    \\
     & =
    \sum_{y \in \LatticeR}
    \sum_{x \in \LatticeR}
    \braket{\varphi}{y}
    \Phi_{y, x}^{(\omega)}
    \braket{x}{\phi}
    \\
     & =
    \bra{\varphi} \left(
    \sum_{y \in \LatticeR}
    \sum_{x \in \LatticeR}
    \Phi_{y, x}^{(\omega)}
    \ket{y}\bra{x}
    \right) \ket{\phi}
\end{align}
The construction extends straightforwardly to bounded linear maps $\Phi: \LtwoR^{\otimes m} \rightarrow \LtwoR^{\otimes n}$, exploiting the isomorphism $\LtwoR^{\otimes n} \cong \LtwoRn$:
\begin{align}
    \delta_{\vec{x}}^{(k)}(\vec{z})
     & := \begin{cases}
        k & \text{ if } z_j \in [\frac{x_j-\frac{1}{2}}{k}, \frac{x_j+\frac{1}{2}}{k}] \text{ for all } j=1,...,n \\
        0 & \text{ otherwise}
    \end{cases}
    \\
    \Phi_{\vec{y}, \vec{x}}^{(k)}
     & :=
    \bra{\delta_{\vec{y}}^{(k)}} \Phi \ket{\delta_{\vec{x}}^{(k)}}
\end{align}
Hence, we can inject bounded linear maps $\LtwoR^{\otimes m} \rightarrow \LtwoR^{\otimes n}$ into linear maps $\LtwoRNonstd^{\otimes m} \rightarrow \LtwoRNonstd^{\otimes n}$ as follows:
\begin{equation}
    \Phi \mapsto
    \sum_{\vec{y} \in \LatticeR^{n}}
    \sum_{\vec{x} \in \LatticeR^{m}}
    \Phi_{\vec{y}, \vec{x}}^{(\omega)}
    \ket{\vec{y}}\bra{\vec{x}}
\end{equation}
By a straightforward generalisation of previous discussion, this injection respects inner products (up to infinitesimals):
\begin{equation}
    \bra{\varphi} \Phi \ket{\phi}
    \simeq
    \bra{\varphi} \left(
    \sum_{\vec{y} \in \LatticeR^{n}}
    \sum_{\vec{x} \in \LatticeR^{m}}
    \Phi_{\vec{y}, \vec{x}}^{(\omega)}
    \ket{\vec{y}}\bra{\vec{x}}
    \right) \ket{\phi}
\end{equation}

\subsection{Bringing Results Back from \texorpdfstring{$\LtwoRNonstd$}{*C[R]} to \texorpdfstring{$\LtwoR$}{L2[R]}}

At this stage, we can lift tensor networks of bras, kets and bounded linear maps from $\LtwoR$ to $\LtwoRNonstd$, and we can reason diagrammatically about them using the full expressive power of the dagger compact structure from $\starHilbCategory$.
After we are done with our calculations, however, we need a way to map our results back down from $\LtwoRNonstd$ to $\LtwoR$.
We only have a hope of doing this when the results make sense in the standard world, i.e. when they are near-standard: in this case, it suffices to take the standard part.
This works for bras, kets, and bounded linear functions:
\begin{align}
    \ket{\varphi} & \mapsto \stdpart{\ket{\varphi}} \\
    \bra{\varphi} & \mapsto \stdpart{\bra{\varphi}} \\
    \Phi          & \mapsto \stdpart{\Phi}
\end{align}
Furthermore, because diagrammatic reasoning progresses by the application of equations, if we start from a near-standard value, one lifted from $\LtwoR$, then we always necessarily end up with a near-standard value, one which can be brought back down to $\LtwoR$.
This is regardless of the building blocks used to represent the inital value, the final value, or anything in between: any tensors $\LtwoRNonstd^{\otimes m} \rightarrow \LtwoRNonstd^{\otimes n}$ can be used as part our reasoning, regardless of whether they are individually near-standard.

\subsection{The Position and Momentum bases}

The ability to work with non-near-standard tensors gives us the freedom to soundly reason using mathematical entities---such as position and momentum eigenstates---which are commonplace in physical calculations but don't have a direct representation as states in in $\LtwoR$.
We have already encountered the (lattice approximation to the) position eigenstates, as the canonical basis for $\LtwoRNonstd$:
\begin{equation}
    \forall x \in \LatticeR.\;
    \ket{x}_X \in \LtwoRNonstd
\end{equation}
We have made the position basis explicit by introducing an $X$ subscript, to distinguish them from the following momentum eigenstates, which are also parametrised by points on the lattice $\LatticeR$:
\begin{equation}
    \label{eq:MomentumBasisDef}
    \forall p \in \LatticeR .\;
    \ket{p}_P := \frac{1}{\omega}\sum_{x \in \LatticeR} e^{i2\pi x p} \ket{x}_X
\end{equation}
The position eigenstates $\ket{x}_X$ form an orthonormal basis by definition.
It is also easy to show that the momentum eigenstates $\ket{p}_P$ form an orthonormal basis:
\begin{align}
    \braket{p}{q}_P
     & =
    \frac{1}{\omega^2}
    \sum_{x \in \LatticeR}
    \sum_{y \in \LatticeR}
    e^{i2\pi (xq-yp)}\;\braket{y}{x}_X
    \\
     & =
    \frac{1}{\omega^2}
    \sum_{x \in \LatticeR}
    e^{i2\pi x(q-p)}
    \\
     & = \begin{cases}
        \frac{1}{\omega^2} \sum\limits_{x \in \LatticeR} 1 & = 1 \text{ if } p = q \\
        \frac{1}{\omega^2} 0                               & = 0 \text{ otherwise}
    \end{cases}
\end{align}
The coefficients of a generic state $\ket{\varphi} = \sum_{x \in \LatticeR} f(x) \ket{x}_X$ in the position and momentum bases are related by the Fourier transform $\mathcal{F}$, which turns a function $f: \LatticeR \rightarrow \starComplexs$ into another function $\hat{f}:= \mathcal{F}(f): \LatticeR \rightarrow \starComplexs$:
\begin{equation}
    \label{eq:FourierTransformBases}
    \sum_{x \in \LatticeR} f(x) \ket{x}_X
    = \sum_{p \in \LatticeR} \hat{f}(p) \ket{p}_P
\end{equation}
The Fourier transform and its inverse are defined as follows:
\begin{align}
    \mathcal{F}(f)
     & := p \mapsto \frac{1}{\omega}\sum_{x \in \LatticeR} e^{-i2\pi xp} f(x)
    \\
    \mathcal{F}^{-1}(\hat{f})
     & := x \mapsto \frac{1}{\omega}\sum_{p \in \LatticeR} e^{i2\pi xp} \hat{f}(p) \label{eq:non-std-inverse-fourier}
\end{align}
As a special case, we obtain the expression of position eigenstates in the momentum basis:
\begin{equation}
    \label{eq:PositionBasisDefMomentum}
    \ket{x}_X = \frac{1}{\omega}\sum_{p \in \LatticeR} e^{-i2\pi xp} \ket{p}_P
\end{equation}
Recalling that $\LatticeR \rightarrow \starComplexs$ and $\LtwoRNonstd$ are identified using the position basis, the Fourier transform defines the following linear map $\mathcal{F}: \LtwoRNonstd \rightarrow \LtwoRNonstd$:
\begin{align}
    \mathcal{F}
     & =
    \sum_{p \in \LatticeR}
    \sum_{x \in \LatticeR}
    \mathcal{F}\left(y \mapsto \delta_{x,y}\right)(p)\ket{p}_X
    \bra{x}_X
    \\
     & =
    \sum_{x \in \LatticeR}
    \sum_{p \in \LatticeR}
    \frac{1}{\omega}e^{-i2\pi xp}\ket{p}_X
    \bra{x}_X
    \\
     & =
    \sum_{x \in \LatticeR}
    \ket{x}_P
    \bra{x}_X
\end{align}
It is immediately seen to be unitary, with inverse $\mathcal{F}^\dagger = \sum_{p \in \LatticeR} \ket{p}_X \bra{p}_P$.

\section{Soundness of Rules}
\label{sec:soundness}
This section is work in progress.
\subsection{Interpretation for the Z and X Spiders}

We can use the two bases above---in the context of the $\starComplexs$-linear, dagger compact structure of $\starHilbCategory$---to define the tensors diagonal in the position and momentum bases, for arbitrary labelling functions $f: \LatticeR \rightarrow \starComplexs$:
\begin{align}
    \interp{\tikzfigscale{1}{z-spider}}
     & := \frac{1}{\omega}
    \sum_{x \in            \LatticeR} f(x) \, \ket{x}^{\otimes n} \bra{x}^{\otimes m}
    \\
    \interp{\tikzfigscale{1}{x-spider-only}}
     & := \frac{1}{\omega}
    \sum_{p \in \LatticeR} f(p) \, \ket{p}^{\otimes n} \bra{p}^{\otimes m}
\end{align}
Because they are defined by orthonormal bases, the spiders form special commutative $\dagger$-Frobenius algebras ($\dagger$-SCFA), with arbitrary phases (not just unitary ones); we refer to them as Z spiders, for the position basis, and X spiders, for the momentum basis.
As a consequence, the \textruleref{Fusion} rules for Z and X spiders are sound in this interpretation:
\begin{equation}
    \tikzfigscale{1}{z-fusion} \hspace{2cm}
    \tikzfigscale{1}{x-fusion}
\end{equation}
Soundness of the \textruleref{Identity} rules for Z and X spiders is also an immediate consequence of the definition:
\begin{equation}
    \tikzfigscale{1}{id-spider-zx}
\end{equation}
Soundness of the \textruleref{Fourier} rules for Z and X spiders follows by definition of the Fourier transform $\mathcal{F}$ (more specifically from Equation \ref{eq:FourierTransformBases}):
\begin{equation}
    \tikzfigscale{1}{fourier-state} \hspace{2cm}
    \tikzfigscale{1}{inverse-fourier-state}
\end{equation}
Soundness of the \textruleref{Copy} rules for Z and X spiders follows again from the definition of the Fourier transform (more specifically, from Equations \ref{eq:MomentumBasisDef} and \ref{eq:PositionBasisDefMomentum}), together with the observation that spiders copy states in their own basis:
\begin{equation}
    \tikzfigscale{1}{z-character-copy} \hspace{2cm}
    \tikzfigscale{1}{x-character-copy}
\end{equation}
Soundness of the \textruleref{Bialgebra} rule for Z and X spiders is related to the Weyl Canonical Commutation Relations and is proven in \cite{Gogioso2018TowardsQFT}:
\begin{equation}
    \tikzfigscale{1}{bialgebra}
\end{equation}
Soundness of the \textruleref{Scalar} rule for Z and X spiders follows from linearity of their definition:
\begin{equation}
    \tikzfigscale{1}{z-scalar} \hspace{2cm}
    \tikzfigscale{1}{x-scalar-fixed}
\end{equation}

For sufficiently well-behaved near-standard functions $f$, the Z and X spiders are themselves near-standard, in which case they can be used to provide non-standard semantics for the following integral expressions:
\begin{align}
    \interp{\tikzfigscale{1}{z-spider}}
     & \simeq \int f(x) \, \ket{x}^{\otimes n} \bra{x}^{\otimes m} dx \\
    \interp{\tikzfigscale{1}{x-spider-only}}
     & \simeq \int f(p) \, \ket{p}^{\otimes n} \bra{p}^{\otimes m} dp
\end{align}
We can check that this is the case on generic product states $\ket{\phi} := \ket{\phi_1} \otimes ... \otimes \ket{\phi_m}$ and $\ket{\varphi} := \ket{\varphi_1} \otimes ... \otimes \ket{\varphi_n}$.
We do so for Z spiders only, as the X spider is analogous:
\begin{align}
    \bra{\varphi}\left(\int f(x) \, \ket{x}^{\otimes n} \bra{x}^{\otimes m} dx \right)\ket{\phi}
     & :=
    \int f(x) \,\varphi_1(x)\cdot ... \cdot \varphi_n(x)\, \phi_1(x)\cdot ... \cdot \phi_m(x) dx
    \\
     & \simeq
    \frac{1}{\omega}\sum_{x \in \LatticeR}
    f(x) \,\varphi_1(x)\cdot ... \cdot \varphi_n(x) \, \phi_1(x)\cdot ... \cdot \phi_m(x)
    \\
     & =
    \bra{\varphi}\left(\frac{1}{\omega}\sum f(x) \, \ket{x}^{\otimes n} \bra{x}^{\otimes m} dx \right)\ket{\phi}
\end{align}

\subsection{The Fock Basis}
\label{appendix:subsec-fock-basis}

The Fock basis is related to the position basis by the Hermite functions:
\begin{equation}
    \ket{n} = \int \psi_n(x) \ket{x} \, dx
\end{equation}
where the Hermite functions $\psi_n(x)$ are themselves defined in terms of the Hermite polynomials $H_n(x)$:
\begin{equation}
    \psi_n(x) := {(2^n n! \sqrt{\pi})}^{-\frac{1}{2}} e^{-x^2/2} H_n(x)
\end{equation}
The Hermite polynomials $H_n(x)$ satisfy the following recursion, starting from $H_0(x) = 1$ and $H_1(x) = 2x$:
\begin{equation}
    H_n(x) = 2xH_{n-1}(x) - 2(n-1)H_{n-2}(x)
\end{equation}
They also have the following closed form expressions:
\begin{equation}
    H_n(x) =
    n! \sum\limits_{m=0}^{\lfloor\frac{n}{2}\rfloor} \frac{
        (-1)^{m}
    }{
        m! (n-2m)!
    } (2x)^{n-2m}
\end{equation}
As the above expression is defined for all $n \in \naturals$, by Transfer Theorem it is also defined for all $n \in \starNaturals$, as are the Hermite functions $\psi_n(x)$.
We can define states $\ket{n} \in \LtwoRNonstd$ for all $n \in \starNaturals$:
\begin{equation}
    \ket{n} := \frac{1}{\omega}\sum_{x \in \LatticeR} \psi_n(x) \ket{x}
\end{equation}
Unfortunately, the non-standard states $\ket{n}$ defined above cannot possibly form an orthonormal basis of $\LtwoRNonstd$: the latter is $\omega^2$-dimensional, so a basis cannot have more than $\omega^2$ states in it.
Indeed, the usual proof of orthonormality for the Fock basis relies on integration and limits, and it does not directly lift to the non-standard model.
Orthonormality, at least for the $n$ in some suitable subset $N_\omega\subset \starNaturals$, would corresponds to the following statement:
\begin{equation}
    \braket{n}{m} = \frac{1}{\omega}\sum_{x \in \LatticeR} \psi_n(x) \psi_m(x)
    = \delta_{n,m}
\end{equation}
For $n, m \in \naturals$, we know that this holds at least up to infinitesimals, but we cannot take $N_\omega := \naturals$ as the set $\naturals$ itself does not exist from the perspective of the non-standard model:
\footnote{
    This was previously mentioned in the context of the subset inclusion $\reals \subset \starReals$, and it is a fundamental concept in non-standard analysis: in order for interesting non-standard models to exist, one must relax the requirement that they include subsets corresponding to all logical statements which can be formulated in the theory.
    Fewer sets, functions and properties available make it possible for non-standard models to be (significantly) larger than their standard counterparts, while satisfying the same axioms.
}
\begin{equation}
    \frac{1}{\omega}\sum_{x \in \LatticeR} \psi_n(x) \psi_m(x)
    \simeq \delta_{n,m} \text{ for } n, m \in \naturals
\end{equation}
However, soundness of the rules involving Fock spiders requires exact orthonormality for all $n$ in a suitable subset $N_\omega$ with exactly $\omega^2$ elements, which can furthermore be expressed algebrically as a function of $\omega$.
In order to obtain the orthonormality statement by Transfer Theorem, we would have to prove, purely with algebraic methods, that the following statement holds in the standard model, for all odd $k \in \naturals$ large enough and all $n \in N_k$, where $N_k$ is a subset of size $k^2$ which can be expressed as a suitable function of $k$:
\begin{equation}
    \delta_{n,m} = \frac{1}{k}\sum_{j = -\sfrac{(k^2-1)}{2}}^{\sfrac{(k^2-1)}{2}}
    \psi_n\left(\sfrac{j}{k}\right)^* \psi_m\left(\sfrac{j}{k}\right)
\end{equation}
Observing that the set $\{-k', ..., k'\}$ is closed under negation for any $k' \in \naturals$, we conclude that the above expression holds for all pairs of distinct $n, m$ where exactly one of $n$ and $m$ is odd, since the function being summed over is odd in that case.
Unfortunately, the discrete Gaussian sums in the case where $n+m$ is even have proven tricky to compute analytically.

We don't expect that the above expression will hold exactly for arbitrary $n, m$, but know that it is possible to express the required corrections explicitly as a function of $k$, defining some $\psi_{n,k}(x) := \alpha_{n,k} \psi_{n}(x) + \gamma_{n, k}(x)$ for $n \in N_k$ in such a way that the following holds exactly:
\begin{equation}
    \delta_{n,m} = \frac{1}{k}\sum_{j = -\sfrac{(k^2-1)}{2}}^{\sfrac{(k^2-1)}{2}}
    \psi_{n,k}\left(\sfrac{j}{k}\right)^* \psi_{m,k}\left(\sfrac{j}{k}\right)
\end{equation}
We also know that it is possible to formulate the required corrections in such a way as to obtain $\psi_{n,\omega} \simeq \psi_{n}$ for all finite $n \in N_\omega$, providing the following non-standard interpretation for the Fock spider:
\begin{align}
    \interp{\tikzfigscale{1}{fock-spider}}
     & := \frac{1}{\omega}
    \sum_{\nu \in N_\omega} g(\nu) \, \ket{\nu}^{\otimes n} \bra{\nu}^{\otimes m}
\end{align}
This will satisfy \textruleref{Fusion} exactly, but soundness of other rules might require some additional infinitesimal corrections. The specific details are left to future work.

\subsection{Relations between Z, X, and Fock Bases}
In the previous two subsections, we discussed soundness of rules involving just Z and X spiders, and of rules involving just the Fock basis.
Here, we discuss soundness of rules involving both of these, which holds in the non-standard setting up to infinitesimal error.

\begin{restatable}{proposition}{eulerSoundness}\label{prop:eulerSoundness}
    \textruleref{Euler} is sound.
\end{restatable}
\begin{proof}
    The phase space rotation operator in CVQC where $\theta \in (\minu \pi, \pi)$
    \begin{equation}
        R(\theta) = e^{\minu i \theta \hat{n}}
    \end{equation}
    has matrix elements~\cite{Tasca2011cvqcspatialdof}
    \begin{equation}
        \bra{k} R(\theta) \ket{m} = \frac{(e^{\minu i\frac{\pi}{4}})^{\frac{\sin\theta}{\abs{\sin\theta}}}}{\sqrt{\abs{\sin\theta}}} e^{i \pi \cot\theta (k^2 \plus{+} m^2)} e^{\minu i 2 \pi \frac{k m}{\sin\theta}}
    \end{equation}
    If $\sin\theta = 0$, then $\theta = 0$ and $\tan \frac{\theta}{2} = 0$, so the right hand side of \textruleref{Euler} does realise $R(0)$.
    If $\sin\theta \neq 0$, we indeed find $R(\minu \theta)$ equal to the right hand side of \textruleref{Euler}:
    \begin{align*}
         & \quad \int e^{i\pi \tan\frac{\theta}{2} k^2} \ket{k}_X \bra{k}_X dk \int e^{i \pi \sin \theta \ell^2} \ket{\ell}_P \bra{\ell}_P d\ell \int e^{i \pi \tan\frac{\theta}{2} m^2} \ket{m}_X \bra{m}_X dm                                                                             \\
         & = \int e^{i\pi (\csc\theta \minu \cot\theta) k^2} \ket{k}_X \bra{k}_X dk \int e^{i \pi \sin \theta \ell^2} \ket{\ell}_P \bra{\ell}_P d\ell \int e^{i \pi (\csc\theta \minu \cot\theta) m^2} \ket{m}_X \bra{m}_X dm                                                               \\
         & = \int e^{i\pi (\csc\theta \minu \cot\theta) k^2} \ket{k}_X \bra{k}_X dk \int e^{i \pi \sin \theta \ell^2} \left(\int e^{i 2 \pi k' \ell} \ket{k'}_X dk'\right)                                                                                                                  \\
         & \quad \left(\int e^{\minu i 2\pi m' \ell} \bra{m'}_X dm'\right) d\ell \int e^{i \pi (\csc\theta \minu \cot\theta) m^2} \ket{m}_X \bra{m}_X dm                                                                                                                                    \\
         & = \iint e^{\minu i \pi \cot\theta (k^2 \plus m^2)} \int e^{i \pi \left(\frac{k^2 \plus m^2}{\sin\theta} \plus \sin\theta \ell^2\right)} e^{i 2 \pi (k\minu m)\ell} d\ell \ket{k}_X \bra{m}_X dk\;dm                                                                              \\
         & = \iint e^{\minu i \pi \cot\theta (k^2 \plus m^2)} e^{i \pi \frac{k^2 \plus m^2}{\sin\theta}} \int e^{i \pi \left( \left(\sqrt{\sin\theta}\ell \plus \frac{k\minu m}{\sqrt{\sin\theta}}\right)^2 \minu \frac{(k\minu m)^2}{\sin\theta} \right)} d\ell \ket{k}_X \ket{m}_X dk\;dm \\
         & = \iint e^{\minu i \pi \cot\theta(k^2 \plus m^2)} e^{i \pi \frac{(k^2 \plus m^2) \minu (k \minu m)^2}{\sin\theta}} \int e^{i \pi \sin\theta \left( \ell \plus \frac{k\minu m}{\sin\theta} \right)^2} d\ell \ket{k}_X \bra{m}_X dk\;dm = (*)
    \end{align*}
    Now we evaluate that
    \begin{equation*}
        \int e^{i \pi \sin\theta \left( \ell \plus \frac{k\minu m}{\sin\theta} \right)^2} d\ell = \int e^{i \alpha \left( \ell \plus \frac{k\minu m}{\sin\theta} \right)^2} d\ell = e^{i \frac{\pi}{4} \frac{\alpha}{\abs{\alpha}}} \sqrt{\frac{\pi}{\abs{\alpha}}}
    \end{equation*}
    where $\alpha = \pi \sin\theta$.  Continuing,
    \begin{align*}
        (*) & = \frac{(e^{i\frac{\pi}{4}})^{\frac{\sin\theta}{\abs{\sin\theta}}}}{\sqrt{\abs{\sin\theta}}} \iint e^{\minu i \pi \cot\theta (k^2 \plus m^2)} e^{i \pi \frac{\left(k^2 \plus m^2\right) \minu (k \minu m)^2}{\sin\theta}} \ket{k}_X \bra{m}_X dk\;dm \\
            & = \frac{(e^{i\frac{\pi}{4}})^{\frac{\sin\theta}{\abs{\sin\theta}}}}{\sqrt{\abs{\sin\theta}}} \iint e^{\minu i \pi \cot\theta (k^2 \plus m^2)} e^{i 2\pi \frac{km}{\sin\theta}} \ket{k}_X \bra{m}_X dk\;dm
    \end{align*}
\end{proof}

\section{Proofs of Lemmas}
\label{sec:rule-proofs-appendix}

\subsection{Derived Rules in Figure~\ref{fig:Derived-Rules}}
\xFusion*
\begin{proof}
    \begin{equation*}
        \tikzfigscale{1}{x-fusion-proof}
    \end{equation*}
\end{proof}

\nul*
\begin{proof}
    \begin{equation*}
        \tikzfigscale{1}{null-proof}
    \end{equation*}
\end{proof}

\WUnit*
\begin{proof}
    \begin{equation*}
        \tikzfigscale{1}{w-unit-proof}
    \end{equation*}
\end{proof}

\vacuumCopy*
\begin{proof}
    \begin{equation*}
        \tikzfigscale{1}{w-vacuum-copy-proof}
    \end{equation*}
\end{proof}

\push*
\begin{proof}
    This is simply a special case of the \textruleref{Bialgebra} rule between the Fock spider and the W node.
\end{proof}

\minusOne*
\begin{proof}
    \begin{equation*}
        \tikzfigscale{1}{minusOne-proof}
    \end{equation*}
\end{proof}

\plusOne*
\begin{proof}
    \begin{equation*}
        \tikzfigscale{1}{plusOne-proof}
    \end{equation*}
\end{proof}

\multiplierIso*
\begin{proof}
    \begin{equation*}
        \tikzfigscale{1}{multiplier-iso-proof}
    \end{equation*}
\end{proof}

\multiplierRev*
\begin{proof}
    \begin{equation*}
        \tikzfigscale{1}{multiplier-rev-proof}
    \end{equation*}
\end{proof}

\multiplierIsoRev*
\begin{proof}
    \begin{equation*}
        \tikzfigscale{1}{multiplier-iso-rev-proof}
    \end{equation*}
\end{proof}

\multiplierAltDecomposition*
\begin{proof}
    \begin{equation*}
        \tikzfigscale{1}{multiplier-alt-decomposition-proof}
    \end{equation*}
\end{proof}

\multiplierCopy*
\begin{proof}
    \begin{equation*}
        \tikzfigscale{1}{multiplier-copy-proof}
    \end{equation*}
\end{proof}
\xmultiplierCopy*
\begin{proof}
    Same as for Lemma~\ref{lem:multiplierCopy}.
\end{proof}

\xTriforce*
\begin{proof}
    \begin{equation*}
        \tikzfigscale{1}{x-triforce-proof}
    \end{equation*}
\end{proof}

\wLoopGreen*
\begin{proof}
    \begin{equation*}
        \tikzfigscale{1}{w-loop-green-proof}
    \end{equation*}
\end{proof}

\KtoNFockZ*
\begin{proof}
    \begin{equation*}
        \tikzfigscale{1}{k-to-n-fock-z-proof}
    \end{equation*}
    where in the last step, we use the exponential generating function of the Hermite polynomials.
\end{proof}

\FWPlusState*
\begin{proof}
    \begin{equation*}
        \tikzfigscale{1}{f-plus-state-proof}
    \end{equation*}
\end{proof}

\FWPlus*
\begin{proof}
    \begin{equation*}
        \tikzfigscale{1}{f-plus-proof}
    \end{equation*}
\end{proof}

\squeezedVacuum*
\begin{proof}
    \begin{equation*}
        \tikzfigscale{1}{squeezed-vacuum-proof}
    \end{equation*}
\end{proof}

\subsection{Proofs for Section~\ref{sec:gates}}
\phat*
\begin{proof}
    We use the fact that $\hat{p} = i\frac{a^\dag \minu a}{\sqrt{2}}$.
    \begin{equation*}
        \tikzfigscale{1}{p-hat-proof}
    \end{equation*}
    Likewise, opening up the 1 number state yields the controlled $\hat{p}$ operator because
    \begin{equation*}
        \tikzfigscale{1}{p-hat-proof-0}
    \end{equation*}
\end{proof}

\fonefact*
\begin{proof}
    \begin{equation*}
        \tikzfigscale{1}{f-1-fact-proof}
    \end{equation*}
    In the third to last step, we used the relation of Hermite functions
    \begin{equation}
        \sqrt{2(n\plus 1)} \psi_{n\plus 1}(x) = \left(x - \frac{d}{dx}\right) \psi_n(x)
    \end{equation}
    to derive
    \begin{align}
        (2x - 1) e^{\minu \frac{x^2}{2} \plus x \minu \frac{1}{4}} & = (x - (\minu x + 1)) e^{\minu \frac{x^2}{2} \plus x \minu \frac{1}{4}}                                                                       \\
                                                                   & = \left(x - \frac{d}{dx}\right) e^{\minu \frac{x^2}{2} \plus x \minu \frac{1}{4}}                                                             \\
                                                                   & = \pi^{\frac{1}{4}} \left(x - \frac{\partial}{\partial x}\right) \sum_{n=0}^\infty 2^{\minu \frac{n}{2}} (n!)^{\minu \frac{1}{2}} \psi_n(x)   \\
                                                                   & = \pi^{\frac{1}{4}} \sum_{n=0}^\infty 2^{\minu \frac{n}{2}} (n!)^{\minu \frac{1}{2}} \sqrt{2(n\plus 1)} \psi_{n \plus 1}(x)                   \\
                                                                   & = \sqrt{2} \pi^{\frac{1}{4}} \sum_{m=1}^\infty 2^{\minu \frac{m\minu 1}{2}} \left((m\minu 1)!\right)^{\minu \frac{1}{2}} \sqrt{m} \psi_{m}(x) \\
                                                                   & = 2 \pi^{\frac{1}{4}} \sum_{m=0}^\infty 2^{\minu \frac{m}{2}} (m!)^{\minu \frac{1}{2}} m \psi_{m}(x)
    \end{align}
\end{proof}

\subsection{Proofs for Gaussian Completeness}
\label{sec:gaussian-completeness-appendix}

\propSemanticsEquivalence*
\begin{proof}
    We first restrict our attention to $\AffLagRel_{\reals}$ because $\AffLagRel_{\complexs}^+$ can be generated from it by adding a vacuum generator which co-discards symplectic rotations and scalars~\cite{boothCompleteGaussian2024}.
    Recall that for a qudit with a power of prime dimension $p$, there is a symmetric-monoidal equivalence between $\AffLagRel_{\mathbb{F}_p}$ and stabilizer circuits for the $p$-dimensional qudit~\cite{grossHudsonTheoremFinitedimensional2006,comfortGraphicalCalculusLagrangian2022}.
    By the transfer theorem, this equivalence extends to the hyperfinite dimensional qudits.
    Hence, we have the equivalence between $\AffLagRel_{\LatticeR}$ and Gaussian operators with position and momentum eigenstates.
    Now, as explained in Section~\ref{sec:non-standard-lattice}, there is an injection from $\reals$ to $\LatticeR$, which can be applied pointwise to go from $\AffLagRel_{\reals}$ to $\AffLagRel_{\LatticeR}$.
    This gives us a functor from $\AffLagRel_{\reals}$ to $\starHilbCategory_G$.

    Next, we extend this functor by sending the vacuum state in $\AffLagRel_{\complexs}^+$ to the following state in $\starHilbCategory_G$.
    \begin{equation}
        \left\{ \left(\bullet,
        \begin{bmatrix}
                i x \\  x
            \end{bmatrix}
        \right) \,\middle|\, x \in \reals \right\}
        \quad \longmapsto \quad
        \frac{1}{\omega} \sum_{x \in \LatticeR} e^{-\frac{x^2}{2}} \ket{x}
    \end{equation}
    It is a straightforward calculation to verify that this state satisfies the equalities required to generate $\AffLagRel_{\complexs}^+$, which correspond to the following graphically.
    \begin{equation*}
        \includegraphics{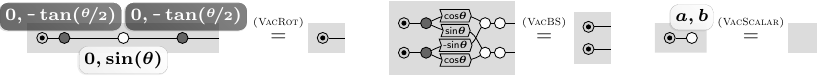}
    \end{equation*}
    This defines a functor from $\AffLagRel_{\complexs}^+$ to $\starHilbCategory_G$.
    This functor is faithful and surjective on objects by construction.
    However, it is not full as there are morphisms in $\starHilbCategory_G$ that correspond to affine Lagrangian relations defined on elements of $\LatticeR$ that are not near-standard.
    By restricting to the subcategory $\starHilbCategory_G^{\text{fin}}$ of Gaussian operators with finite parameters, we can eliminate such morphisms.
    Thus, we obtain a functor $\AffLagRel_{\complexs}^+ \to \starHilbCategory_G^{\text{fin}}$ that is full, faithful and essentially surjective on objects.
    Hence, the category $\starHilbCategory_G^{\text{fin}}$ is equivalent to $\AffLagRel_{\complexs}^+$.
\end{proof}

\subsubsection{Deriving axioms of the graphical symplectic algebra}

\begin{lemma}\label{lem:gsa-lem-zero}
    \begin{equation}
        \tikzfigscale{1}{gsa-lem-zero}
    \end{equation}
\end{lemma}
\begin{proof}
    \begin{equation*}
        \tikzfigscale{1}{gsa-lem-zero-proof}
    \end{equation*}
\end{proof}

\begin{restatable}{lemma}{hboxFusion}\label{lem:hboxFusion}
    \begin{equation}
        \tikzfigscale{1}{hbox-fusion}
    \end{equation}
\end{restatable}
\begin{proof}
    \begin{equation*}
        \tikzfigscale{1}{hbox-fusion-proof}
    \end{equation*}
\end{proof}
\begin{lemma}\label{lem:gsa-one-1}
    \begin{equation}
        \tikzfigscale{1}{gsa-one-1}
    \end{equation}
\end{lemma}
\begin{proof}
    \begin{equation*}
        \tikzfigscale{1}{gsa-one-1-proof}
    \end{equation*}
\end{proof}
\begin{lemma}\label{lem:gsa-one-2}
    \begin{equation}
        \tikzfigscale{1}{gsa-one-2}
    \end{equation}
\end{lemma}
\begin{proof}
    \begin{equation*}
        \tikzfigscale{1}{gsa-one-2-proof}
    \end{equation*}
\end{proof}
\begin{lemma}\label{lem:gsa-vacuum-rot}
    \begin{equation}
        \tikzfigscale{1}{gsa-vacuum-rot}
    \end{equation}
\end{lemma}
\begin{proof}
    \begin{equation*}
        \tikzfigscale{1}{gsa-vacuum-rot-proof}
    \end{equation*}
\end{proof}
\begin{lemma}\label{lem:gsa-vacuum-bs}
    \begin{equation}
        \tikzfigscale{1}{gsa-vacuum-bs}
    \end{equation}
\end{lemma}
\begin{proof}
    \begin{equation*}
        \tikzfigscale{1}{gsa-vacuum-bs-proof}
    \end{equation*}
\end{proof}
\begin{lemma}\label{lem:gsa-vacuum-scalar}
    \begin{equation}
        \tikzfigscale{1}{gsa-vacuum-scalar}
    \end{equation}
\end{lemma}
\begin{proof}
    \begin{equation*}
        \tikzfigscale{1}{gsa-vacuum-scalar-proof}
    \end{equation*}
\end{proof}

\subsection{Proofs for Gaussian Boson Sampling}
\label{sec:gbs-appendix}
\Knloops*
\begin{proof}
    \begin{equation*}
        \tikzfigscale{1}{complete-graph-with-loops-proof}
    \end{equation*}
\end{proof}

\Knouter*
\begin{proof}
    \begin{equation*}
        \tikzfigscale{1}{complete-graph-outer-product-proof}
    \end{equation*}
\end{proof}

\matrixadd*
\begin{proof}
    \begin{equation*}
        \tikzfigscale{1}{matrix-addition-proof}
    \end{equation*}
\end{proof}

\loopPop*
\begin{proof}
    \begin{equation*}
        \tikzfigscale{1}{loop-pop-proof}
    \end{equation*}
\end{proof}


\loopPopGen*
\begin{proof}
    \begin{equation*}
        \tikzfigscale{1}{loop-pop-gen-proof}
    \end{equation*}
\end{proof}

\end{document}